\documentclass[journal]{IEEEtran}
\usepackage{ifpdf}

\usepackage{cite}

\ifCLASSINFOpdf
\usepackage[pdftex]{graphicx}
\else
\fi
\usepackage{amsmath}
\usepackage{amssymb}
\usepackage{amsthm}
\usepackage{algorithmic}

\usepackage{array}
\usepackage{multirow}
\usepackage{tabularx}
\usepackage{booktabs}

\ifCLASSOPTIONcompsoc
  \usepackage[caption=false,font=normalsize,labelfont=sf,textfont=sf]{subfig}
\else
  \usepackage[caption=false,font=footnotesize]{subfig}
\fi
\usepackage{fixltx2e}
\usepackage{dblfloatfix}
\usepackage{enumitem}

\ifCLASSOPTIONcaptionsoff
 \usepackage[nomarkers]{endfloat}
\let\MYoriglatexcaption\caption
\renewcommand{\caption}[2][\relax]{\MYoriglatexcaption[#2]{#2}}
\fi
\usepackage{url}

\begin{document}
%
\title{Variable-Step Time-Delay Control for Proactive Aperiodic Spacecraft Attitude Control}
%
%
%

\author{Yonsoo~Kim and Hancheol~Cho%
\thanks{Y. Kim is with the Department of Astronomy, Yonsei University, Seoul, South Korea (e-mail: yonsooh@yonsei.ac.kr).}%
\thanks{H. Cho is with the Department of Aeronautics and Astronautics and the Department of Satellite Systems, Yonsei University, Seoul, South Korea (e-mail: hancho37@yonsei.ac.kr).}%
\thanks{This manuscript is a revised preprint prepared for submission to the \textit{IEEE Transactions on Aerospace and Electronic Systems}.}}
\maketitle

\begin{abstract}
This paper addresses proactive aperiodic spacecraft attitude control under model uncertainty, environmental disturbances, and actuator degradation. To this end, we develop Variable-step Time Delay Control (VTDC), which jointly designs robust feedback control and control-update scheduling rather than treating them as separate components. Inspired by adaptive stepsize regulation in Runge–Kutta integration, VTDC structures the local control error to scale with the realized timestep. Using sliding-manifold-based Time Delay Control, the resulting System Time Delay Error (TDE), which reflects local model and uncertainty variations, is shown to be quadratically bounded by the control interval. This relation yields a closed-form feedback law that enlarges or reduces the subsequent interval to regulate the TDE magnitude. The next update time is therefore determined algebraically without continuous trigger monitoring, future-state prediction, or iterative search. The resulting variable-step closed loop admits bounded timesteps and step ratios, excludes Zeno behavior, and renders the sliding variable uniformly ultimately bounded. Nonlinear spacecraft attitude-control simulations demonstrate accurate tracking with low scheduling cost under representative uncertainties and disturbances.

\end{abstract}


%

\section{Introduction}
%
%
%
%
\IEEEPARstart{S}{pacecraft} dynamics evolve continuously in time, whereas onboard control is necessarily implemented at discrete update instants due to limitations in onboard computation, sensor update rates, communication, and actuator capabilities. Since the control input is typically maintained by a zero-order hold between sampling instants, the choice of the control interval directly affects closed-loop stability and tracking performance. If the control interval is too long, the control input may fail to sufficiently follow the system evolution, resulting in performance degradation or even instability. Conversely, if the control interval is too short, practical implementation issues such as increased computational burden, sensor noise, actuator bandwidth and saturation, and finite word length effects may become significant. Therefore, determining an appropriate control interval is a key issue in the design of digital control systems \cite{franklin1998digital}, \cite{wittenmark2002computer}.

In conventional sampled-data control studies, the inter-sample behavior induced by the zero-order hold has been interpreted as a time-varying input delay, and the maximum allowable sampling interval that guarantees closed-loop stability has been estimated using Lyapunov–Krasovskii functionals and linear matrix inequality (LMI)-based conditions \cite{fridman2014introduction}. However, such methods generally provide a uniform maximum allowable sampling interval that is commonly applied to all sampling sequences, making it difficult to directly reflect the instantaneous admissible control interval, which may vary depending on the system state, disturbance level, operating condition, and model uncertainty. Moreover, LMI-based stability conditions are often sufficient conditions and may therefore yield conservative values compared with the actual control interval that the system can tolerate\cite{fridman2005input}, \cite{seuret2012novel}.

To alleviate these limitations, aperiodic control methods have been studied as alternatives to time-triggered control, in which the update times of the control input are determined according to the system state or performance requirements. Representative approaches include Event-Triggered Control (ETC) and Self-Triggered Control (STC) \cite{heemels2012introduction}. Both methods can be understood as combining a general feedback control law with a triggering mechanism that determines when the control input should be updated. ETC is a reactive approach that updates the control input when a triggering condition is violated based on measurements, whereas STC is a proactive approach that computes the next control time in advance using the current state and system model \cite{heemels2012introduction}. 

Satellites operate under multiple constraints, including limited onboard computational resources, power consumption of sensors and actuators, communication bandwidth constraints with ground stations or between satellites \cite{fortescue2011spacecraft}. Under periodic control, an update rate selected to accommodate demanding operating conditions is maintained even when the system evolves slowly, which may lead to unnecessary sensing, computation, and control updates. Aperiodic control can alleviate this inefficiency by adjusting the update instants according to the instantaneous control demand. For this reason, ETC-based control methods have recently been investigated in various spacecraft problems, including attitude control, attitude tracking, attitude synchronization of multiple satellites, and orbit maintenance \cite{di2021event}-\cite{xie2024dynamic}.  However, although ETC can reduce unnecessary control updates, it requires continuous monitoring of the system state or output to determine whether the triggering condition is satisfied. 

In contrast, STC determines the next sensing and control time in advance, thereby reducing sensing, computation, and communication burdens, and can thus be a more suitable alternative for resource-constrained satellite control systems \cite{heemels2012introduction}. Nevertheless, important limitations remain when applying STC to nonlinear spacecraft attitude control problems. Since STC must predict the evolution of future states or triggering errors, prediction errors directly affect stability guarantees in the presence of model uncertainties and disturbances. To compensate for this issue, existing studies often rely on disturbance bounds, uncertainty sets, reachable sets, or conservative calculations of the inter-execution time \cite{de2020self}. In addition, for general nonlinear systems, the next triggering time is difficult to compute in closed form, and numerical integration, iterative search, or conservative bound calculations may be required. Although several attempts have been made to address these issues, some approaches rely on the homogeneity of nonlinear systems, which can become problematic in practical systems that may become nonhomogeneous due to disturbances and other effects \cite{anta2010sample}, \cite{delimpaltadakis2021region}. These characteristics limit their applicability to realistic spacecraft attitude control problems, where actuator saturation, environmental disturbances, model uncertainties, and variations in operating conditions coexist.

To address these issues, this study proposes a proactive aperiodic robust control framework that combines Sliding Mode Control (SMC) and Time Delay Control (TDC). SMC is a representative nonlinear control method that structures tracking error dynamics through a sliding manifold and provides robustness against disturbances and model uncertainties. However, conventional SMC suffers from chattering caused by discontinuous switching inputs, which can induce actuator wear and excite unmodeled high-frequency dynamics \cite{utkin2013sliding}. To mitigate these limitations, Cho et al. proposed a smooth adaptive robust control structure for satellite formation flying, showing that chattering can be eliminated and trajectory errors can be maintained within a user-specified bound without prior knowledge of the upper bounds of disturbances and uncertainties \cite{cho2020autonomous}.

Building on this smooth robust control perspective, the current study combines the sliding-manifold-based error dynamics of SMC with the data-driven uncertainty compensation of TDC. TDC estimates and compensates for lumped uncertainty using input and state information from the previous time instant, thereby maintaining robustness even in control environments with significant disturbances and model uncertainties \cite{youcef1990time}. In particular, since the Time-Delay Error (TDE) generated in TDC directly depends on the control interval, this study utilizes TDE not merely as an estimation error but as a key indicator for determining the next sensing and control time. Accordingly, this study designs Variable-step Time Delay Control (VTDC), which actively adjusts the control interval according to the current state and the estimated TDE level without directly predicting future state trajectories or continuously monitoring triggering conditions.

The proposed VTDC  is therefore applied to proactive aperiodic spacecraft attitude control, where the required control effort and update rate can vary with maneuvering conditions, disturbances, and system uncertainties. VTDC adapts the control interval according to the local System TDE while maintaining robust tracking under model uncertainties, disturbances, and actuator degradation. By linking attitude-control performance and update scheduling through the TDE, the framework allocates control updates according to the local control demand while avoiding continuous triggering and future-state prediction.
The main contributions of this study are summarized as follows:

\begin{itemize} [label=$-$]
    \item (TDE-Based Control--Scheduling Co-Design) The System TDE is reinterpreted as an interval-dependent indicator of local model and disturbance variations and is directly incorporated into control scheduling. By establishing its quadratic dependence on the realized timestep, a closed-form feedback law is derived to proactively adapt the subsequent control interval, thereby integrating robust TDC and aperiodic scheduling within a unified framework.

    \item (Guaranteed Variable-Step Robustness and Stability) The proposed timestep law is analytically shown to attenuate the influence of the initial timestep, maintain the timestep and consecutive step ratio within finite asymptotic bounds, and exclude Zeno behavior. These scheduling properties are further incorporated into the variable-step closed-loop analysis to establish uniform ultimate boundedness of the sliding variable and its steady-state convergence radius.
    
    \item (Prediction-Free Lightweight Scheduling) A practical System-TDE surrogate is constructed from the equivalent and robust control-input variations, enabling online timestep adaptation without direct disturbance measurement, a separate disturbance observer, future-state prediction, or iterative trigger-time search. The online scheduling law therefore requires only locally available control and state information rather than explicit uncertainty-bound information.

    \item (Resource-Efficient Spacecraft Validation) The proposed framework is validated for nonlinear spacecraft attitude maneuvering under inertia uncertainty, environmental disturbances, and time-varying actuator degradation. Comparative simulations with fixed-step TDC and Lyapunov-based self-triggered control demonstrate that VTDC maintains competitive tracking performance while reducing control-update and timestep-scheduling burdens, supporting its applicability to resource-constrained onboard spacecraft control.
\end{itemize}

\theoremstyle{definition}
\newtheorem{remark}{Remark}
\newtheorem{corollary}{Corollary}


\section{Variable-step Time Delay Control Framework}\label{sec:framework}
This section presents the formulation and theoretical development of the proposed VTDC framework. 
Fig.~\ref{fig:VTDCFramework} summarizes the logical structure of the analysis, from robust controller design and TDE-based error characterization to timestep adaptation, closed-loop stability, and practical implementation.

\begin{figure*}[!t]
    \centering
    \includegraphics[width=\textwidth]{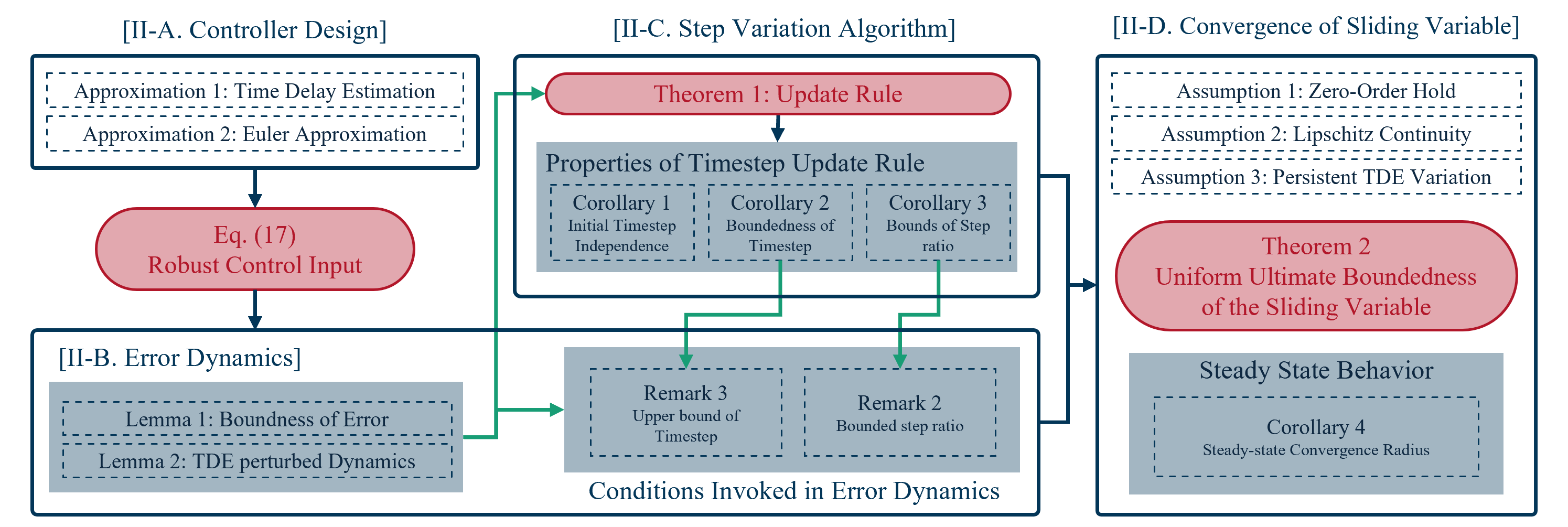}
    \caption{Overall structure and logical dependencies of the proposed VTDC framework, from robust controller design and TDE-based error analysis to timestep adaptation, closed-loop stability, and practical implementation.}
    \label{fig:VTDCFramework}
\end{figure*}

\subsection{Controller Design}
The primary objective of the controller is to ensure that the tracking error $\boldsymbol{e} \in \mathbb{R}^N$ for states such as attitude, position, angular velocity, and velocity converges to zero ($\boldsymbol{e} \to \boldsymbol{0}$) as $t \to \infty$. To describe the error behavior over time, we define the state vector as follows:

\begin{equation}
    \boldsymbol{x} = \begin{bmatrix}\boldsymbol{e}\\\dot{\boldsymbol{e}}\end{bmatrix} \in \mathbb{R}^{2N} \label{eq:StateVector}
\end{equation}
where the dot notation denotes the time derivative of a variable. The governing dynamics of this error state can be expressed as follows:
\begin{equation}
\dot{\boldsymbol x}(t)=\boldsymbol f(\boldsymbol x,t)+\boldsymbol g[\boldsymbol u(t)+\boldsymbol d(t)] \label{eq:Dynamics}
\end{equation}
In this context, $\boldsymbol f(\boldsymbol x,t)\in\mathbb{R}^{2N}$ represents the known nominal nonlinear dynamics of the system, while $\boldsymbol u(t) \in \mathbb{R}^N$ denotes the control input to be designed. The constant matrix $\boldsymbol g \in \mathbb{R}^{2N\times N}$ maps the control action into the state space.

For the considered fully actuated $N$-DOF dynamic system, the generalized control input enters the acceleration-error dynamics through the positive-definite constant inertia matrix $\boldsymbol M\in\mathbb{R}^{N\times N}$, which corresponds to the nominal spacecraft moment-of-inertia matrix $\boldsymbol{J}_0$ in the attitude-control application of Section~\ref{sec:NumDef}. Consequently, the input distribution matrix is given by:
\begin{equation}
\boldsymbol{g}=\begin{bmatrix} \boldsymbol{0}_{N\times N}\\ \boldsymbol{M}^{-1} \end{bmatrix}\label{eq:MappingMat}
\end{equation}
The term $\boldsymbol d(t) \in \mathbb{R}^N$ denotes the matched lumped uncertainty in the generalized input coordinates; for spacecraft attitude control, it represents torque-equivalent effects of environmental disturbances, inertia mismatch, and actuator degradation. In sliding mode control, only uncertainties satisfying the matching condition can be compensated by the control action. That is, the uncertainty must enter the system through the same distribution channel as the control input \cite{utkin2013sliding}. According to Drazenovic, any disturbance satisfying the matching condition can be represented through the same distribution matrix $\boldsymbol{g}$ as the control input \cite{drazenovic1969}. That is, for an arbitrary disturbance $\boldsymbol w(t)$ entering the system as $\boldsymbol g\boldsymbol u(t)+\boldsymbol D\boldsymbol w(t)$, if $\mathrm{Range}(\boldsymbol D)\subseteq \mathrm{Range}(\boldsymbol g)$, then there exists an equivalent lumped uncertainty $\boldsymbol d(t)$ such that $\boldsymbol D\boldsymbol w(t)=\boldsymbol g\boldsymbol d(t)$. Hence, the matched disturbance can be written in the control channel as $\boldsymbol g[\boldsymbol u(t)+\boldsymbol d(t)]$.

In this paper, a sliding variable $\boldsymbol{s}(t) \in \mathbb{R}^N$ is introduced as a linear combination of the error states:
\begin{equation}
\boldsymbol{s}(t) = \boldsymbol{Cx}(t) \label{eq:SlidingVariable}
\end{equation}
where $\boldsymbol{C}\in\mathbb{R}^{N\times2N}$ is a constant gain matrix chosen so that the reduced-order dynamics on the sliding manifold $\boldsymbol{s}(t)=\mathbf{0}$ are stable, thereby ensuring the primary objective $\boldsymbol{e} \to \mathbf{0}$ as $t \to \infty$. Accordingly, $\boldsymbol C$ is chosen as:
\begin{equation}
    \boldsymbol C=\begin{bmatrix}\boldsymbol{B} \ &\boldsymbol I_{N\times N}\end{bmatrix}\label{eq:SVMat}
\end{equation}
where $\boldsymbol B=\mathrm{diag}(\beta_1,\ldots,\beta_N)\in\mathbb{R}^{N\times N}$  with $\beta_i>0$ is selected to make $-\boldsymbol B$  Hurwitz. For the canonical second-order state representation in \eqref{eq:StateVector}, this choice yields $\dot{\boldsymbol e}=-\boldsymbol B\boldsymbol e$ on the sliding manifold, thereby guaranteeing exponential convergence of the tracking error \cite{utkin2013sliding}.

The fundamental principle of sliding mode control consists of two distinct phases: the \textit{reaching phase}, where the controller drives the system state from any initial condition toward the sliding manifold ($\boldsymbol{s}(t) \to \mathbf{0}$), and the \textit{sliding phase}, where the state is confined to the manifold and slides toward the origin robustly against matched uncertainties. To analyze how the system state is driven toward this manifold during the reaching phase, we define the reaching dynamics by taking the time derivative of the sliding variable using \eqref{eq:Dynamics}:

\begin{equation}
\dot{\boldsymbol s}(t)=\boldsymbol C\dot{\boldsymbol x}(t)=\boldsymbol C\Big[\boldsymbol f(\boldsymbol x,t)+\boldsymbol g[\boldsymbol u(t)+\boldsymbol d(t)]\Big] \label{eq:ReachingDynamics}
\end{equation}

The control strategy aims to drive the system to the sliding manifold ($\boldsymbol{s}(t) = \mathbf{0}$) and maintain it there. Once the system reaches the manifold, invariance of the sliding motion requires $\dot{\boldsymbol{s}}(t) = \mathbf{0}$. To achieve this in the presence of uncertainties, the control input is typically derived by first considering an ideal nominal case, followed by the addition of a robust compensation term.

\subsubsection{Equivalent Control}
The equivalent control, $\boldsymbol{u}^{eq}(t)$, represents the ideal control action required to maintain the sliding condition when the lumped uncertainty is neglected \cite{utkin1977}. By assuming a nominal system where $\boldsymbol{d}(t) = \mathbf{0}$, the reaching dynamics in \eqref{eq:ReachingDynamics} can be simplified as follows:
\begin{equation}
\dot{\boldsymbol s}(t)=\boldsymbol C \dot{\boldsymbol x}(t)=\boldsymbol C[\boldsymbol f(\boldsymbol x,t)+\boldsymbol g\boldsymbol u(t)] \label{eq:ReachingDynamicswoD}
\end{equation}

The equivalent control $\boldsymbol{u}^{eq}(t)$ is defined as the control input $\boldsymbol{u}(t)$ that satisfies the condition $\dot{\boldsymbol{s}}(t) = \mathbf{0}$ in the nominal dynamics of \eqref{eq:ReachingDynamicswoD}. From this condition, $\boldsymbol{u}^{eq}(t)$ is analytically derived as:

\begin{equation}
\boldsymbol u^{eq}(t)=-(\boldsymbol{Cg})^{-1}\boldsymbol{C}\boldsymbol{f}(\boldsymbol{x},t) \label{eq:EquivalentControl}
\end{equation}

\begin{remark}
From \eqref{eq:MappingMat} and \eqref{eq:SVMat}, the product $\boldsymbol C\boldsymbol g$ is given by
\begin{equation}
\boldsymbol {Cg}\boldsymbol = \begin{bmatrix} \boldsymbol B & \boldsymbol I_{N\times N} \end{bmatrix}
\begin{bmatrix} \boldsymbol 0_{N\times N}\\ \boldsymbol M^{-1} \end{bmatrix}
= \boldsymbol M^{-1}.    
\end{equation}
Since the inertia matrix $\boldsymbol M$ of a dynamic system is symmetric positive definite, it is nonsingular. Therefore, $\boldsymbol C\boldsymbol g=\boldsymbol M^{-1}$ is also nonsingular, and $(\boldsymbol C\boldsymbol g)^{-1}$ is well-defined.
\end{remark}

\subsubsection{Time Delay Control}
\newtheorem{assumption}{Assumption}
\newtheorem{approximation}{Approximation}
\newtheorem{lemma}{Lemma}
To compensate for the lumped uncertainty $\boldsymbol{d}(t)$ that remains in the reaching dynamics \eqref{eq:ReachingDynamics}, the total control input is partitioned into the nominal part and the robust compensation part: $\boldsymbol{u}(t) = \boldsymbol{u}^{eq}(t) + \boldsymbol{u}^{d}(t)$. By substituting this control law into \eqref{eq:ReachingDynamics} and utilizing the definition of $\boldsymbol{u}^{eq}(t)$ in \eqref{eq:EquivalentControl}, the known nominal dynamics $\boldsymbol{f}(\boldsymbol{x},t)$ are effectively canceled. This yields the following reduced dynamics:

\begin{equation}
\dot {\boldsymbol s} (t) = \boldsymbol{Cg}\big[\boldsymbol{u}^{d}(t)+ \boldsymbol d(t)\big]  \label{eq:ReducedDynamics}
\end{equation}

In practice, direct computation of $\boldsymbol{u}^{d}(t)$ via \eqref{eq:ReducedDynamics} is challenging because the instantaneous values of the disturbance $\boldsymbol{d}(t)$ and the time derivative of the sliding variable $\dot{\boldsymbol{s}}(t)$ are typically unavailable. To address this, the concept of the Time Delay Control (TDC) technique is employed to estimate these terms using past data  \cite{youcef1990time}.

\begin{approximation}[\textit{Time Delay Estimation}]\label{approximation:TimeDelay}
Assuming that the lumped uncertainty $\boldsymbol{d}(t)$ varies slowly relative to a control interval $h$, the current disturbance can be estimated using input and state/output information from the previous step. Based on \eqref{eq:ReducedDynamics}, $\boldsymbol{d}(t)$ is estimated as follows:

\begin{equation}
\boldsymbol d(t) \approx \boldsymbol d(t-h)= (\boldsymbol{Cg})^{-1}\dot{\boldsymbol s}(t-h)-\boldsymbol u^d(t-h) \label{eq:TimeDelayEstimation}
\end{equation}    
\end{approximation}

Substituting the estimated disturbance from \eqref{eq:TimeDelayEstimation} into the control law \eqref{eq:ReducedDynamics} yields the current robust control input:

\begin{equation}
\boldsymbol{u}^{d}(t) = \boldsymbol{u}^{d}(t-h)+(\boldsymbol{Cg})^{-1}[\dot {\boldsymbol s} (t)-\dot {\boldsymbol s} (t-h)]\label{eq:ud_prev}
\end{equation}

\begin{approximation}[\textit{Euler Approximation}]\label{approximation:Euler}
To facilitate discrete-time implementation, the derivative terms $\dot{\boldsymbol s}(t)$ and $\dot{\boldsymbol s}(t-h)$ are approximated using the forward difference method:\cite{franklin1998digital}
\begin{equation}
\dot{\boldsymbol s}(t)\approx\frac{\boldsymbol s(t+h)-\boldsymbol s(t)}{h} 
,\;\;\dot{\boldsymbol s}(t-h)\approx\frac{\boldsymbol s(t)-\boldsymbol s(t-h)}{h} 
\label{eq:EulerApproximation}
\end{equation}

\end{approximation}
In discrete-time sliding mode control, unlike the continuous-time case, the system state cannot reside exactly on the switching manifold $\boldsymbol{s}(t) = \mathbf{0}$ due to the sampling effect; instead, it exhibits a characteristic oscillation within a small neighborhood known as the quasi-sliding mode band \cite{milosavljevic1985}. To ensure stable convergence to this band, we design the predicted sliding variable $\boldsymbol{s}(t+h)$ to follow a specific reaching law as \eqref{eq:ControlObjective}, consistent with the control objective proposed by Gao et al. \cite{gao1995discrete}. By substituting this into \eqref{eq:EulerApproximation}, the control objective at the derivative level can also be calculated as shown in \eqref{eq:ControlObjective2}.
\begin{align}
\boldsymbol s(t+h) &= \gamma\boldsymbol s(t), \qquad\qquad |\gamma|<1\label{eq:ControlObjective}\\
\dot {\boldsymbol s} (t) &= (\gamma-1)\frac{\boldsymbol s(t)}{h}\label{eq:ControlObjective2}
\end{align}
where $\gamma$ is a design parameter that determines the convergence rate. 
Finally, by substituting \eqref{eq:EulerApproximation} and \eqref{eq:ControlObjective} into \eqref{eq:ud_prev} the $\boldsymbol{u}^{d}(t)$ is formulated as follows: 
\begin{equation}
\boldsymbol{u}^{d}(t) = \boldsymbol{u}^{d}(t-h)+ (\boldsymbol{Cg})^{-1}\Big[ (\gamma-1)\frac{\boldsymbol s(t)}{h}- \frac{\boldsymbol s(t)-\boldsymbol s(t-h)}{h} \Big]\label{eq:TimeDelayControl}
\end{equation}

In summary, the total control input $\boldsymbol{u}(t)$ combines the model-based nominal action ($\boldsymbol{u}^{eq}$) and the data-driven robust compensation ($\boldsymbol{u}^{d}$):
\begin{equation}
\begin{split}
    \boldsymbol{u}(t) &= \\ &\underbrace{\boldsymbol{u}^{eq}(t)}_{-(\boldsymbol{Cg})^{-1}\boldsymbol{Cf}(\boldsymbol{x}(t),t)} + \underbrace{\boldsymbol{u}^d(t)}_{\boldsymbol{u}^d(t-h) + (\boldsymbol{Cg})^{-1}\left[(\gamma-1)\frac{\boldsymbol{s}(t)}{h} - \frac{\boldsymbol{s}(t)-\boldsymbol{s}(t-h)}{h}\right]}\label{eq:TotalControl}
\end{split}
\end{equation}

\subsection{Error Dynamics} \label{sec:ErrorDynamics}
\newtheorem{definition}{Definition}
Although the control input designed based on the aforementioned approximations demonstrates robust control performance, it inherently involves errors originating from these two approximations. These errors are particularly sensitive to the sampling interval $h$. To investigate the analytical relationship of these errors, the total control law \eqref{eq:TotalControl} is substituted into the reaching dynamics in \eqref{eq:ReachingDynamics}.

\begin{equation}
    \dot {\boldsymbol s} (t) = \Big[ (\gamma-1)\frac{\boldsymbol s(t)}{h}- \frac{\boldsymbol s(t)-\boldsymbol s(t-h)}{h} \Big]+\boldsymbol{Cg}\big[\boldsymbol{u}^{d}(t-h)+\boldsymbol d(t)\big] 
    \label{eq:ErrorDynamicsPrev1}
\end{equation}
From \eqref{eq:ReducedDynamics}, the term $\boldsymbol{Cg}\boldsymbol{u}^{d}(t-h)$ can be rearranged as follows:
\begin{equation}
    \boldsymbol{Cg}\boldsymbol{u}^{d}(t-h)=\dot {\boldsymbol s} (t-h)-\boldsymbol{Cg}\boldsymbol{d}(t-h)\label{eq:ReducedDynamics_shift}
\end{equation}
By substituting \eqref{eq:ReducedDynamics_shift} into \eqref{eq:ErrorDynamicsPrev1}, we obtain the following.

\begin{equation}
\begin{split}
    \dot {\boldsymbol s} (t) = &\overbrace{(\gamma-1)\frac{\boldsymbol s(t)}{h}}^{\text{Control Objective from \eqref{eq:ControlObjective2}}} \\&+\underbrace{\boldsymbol{Cg}\big[\boldsymbol d(t)-\boldsymbol d(t-h)\big]}_{\text{Error from Approximation \ref{approximation:TimeDelay}}}+\underbrace{\Big[\dot {\boldsymbol s} (t-h) - \frac{\boldsymbol s(t)-\boldsymbol s(t-h)}{h} \Big]}_{\text{Error from Approximation \ref{approximation:Euler}} }
\end{split}\label{eq:raw_ED}
\end{equation}

The two error terms in \eqref{eq:raw_ED}, induced by Approximation \ref{approximation:TimeDelay} and Approximation \ref{approximation:Euler}, can be interpreted as local time-delay errors over the interval $[t-h,t)$. While the first term is already expressed as a time-delay difference of the disturbance, the second term is not yet written in terms of the nominal dynamics and disturbance variations. 

Before rewriting the Euler residual using a local Taylor expansion, we specify the intersample control behavior. Consistent with digital onboard implementation, VTDC employs a zero-order hold (ZOH) \cite{franklin1998digital}, such that the control input is updated at each sampling instant and held constant until the subsequent update.

\begin{assumption}[\textit{Interval-wise Zero-Order Hold}]
\label{assumption:ZOH}
Let $\{t_n\}_{n=0}^{\infty}$ denote the sequence of sampling instants. At each $t_n$, the control input is updated and its post-update value is defined as $\boldsymbol u_n \triangleq \boldsymbol u(t_n^+)$. Under the ZOH implementation,
\begin{equation}
    \boldsymbol u(t)=\boldsymbol u_n, \qquad\forall t\in[t_n,t_{n+1}).
\end{equation}
Consequently, $\dot{\boldsymbol u}(t)$ is given by \eqref{eq:ZeroOrderHold}, whereas $\boldsymbol u(t)$ may exhibit discontinuities at the sampling instants.
\begin{equation}
    \dot{\boldsymbol u}(t)=\boldsymbol 0, \qquad \forall t\in(t_k,t_{k+1}),
    \label{eq:ZeroOrderHold}
\end{equation}
\end{assumption}

Because the control input may undergo a jump at $t_k$, its classical time derivative need not exist exactly at the sampling instant. Moreover, since $\dot{\boldsymbol s}$ depends directly on $\boldsymbol u$ through \eqref{eq:ReachingDynamics}, $\dot{\boldsymbol s}$ may also be discontinuous at $t_k$. 
Therefore, throughout the following local intersample analysis, derivatives evaluated at the left endpoint of a sampling interval are understood as right-hand derivatives. For example, $\dot{\boldsymbol s}(t_k)$ and $\ddot{\boldsymbol s}(t_k)$ represent $\dot{\boldsymbol s}(t_k^+)$ and $\ddot{\boldsymbol s}(t_k^+)$, respectively. For notational simplicity, the superscript $(\cdot)^+$ is omitted hereafter whenever this convention is unambiguous.

Under this convention, we rewrite the Euler residual by expanding $\boldsymbol{s}(t)$ about the previous control instant $t-h$ using a Taylor expansion \cite{butcher2016numerical}:
\begin{equation} 
\boldsymbol s (t) = \boldsymbol s (t-h) + h\dot{\boldsymbol s} (t-h) + \frac{h^2}{2}\ddot{\boldsymbol s} (t-h) +\mathcal{O}(h^3). \label{eq:sEuler} 
\end{equation} 
Subtracting $\boldsymbol s(t-h)$ from both sides and dividing by $h$ yields
\begin{equation} 
\dot {\boldsymbol s} (t-h) - \frac{\boldsymbol s(t)-\boldsymbol s(t-h)}{h} = -\frac{h}{2}\ddot{\boldsymbol s} (t-h)+\mathcal{O}(h^2)\label{eq:RawModel} 
\end{equation}

To characterize $\ddot{\boldsymbol s}$ within an intersample
interval, we differentiate the reaching dynamics in
\eqref{eq:ReachingDynamics}:
\begin{equation}
\begin{split}
\ddot{\boldsymbol s}(t) &= \frac{d}{dt}\Bigg(\boldsymbol{C}\Big[\boldsymbol f(\boldsymbol x,t)+\boldsymbol{g}\big[\boldsymbol u(t)+\boldsymbol d(t)\big]\Big]\Bigg)\\
&=\boldsymbol{C}\Big[\dot{\boldsymbol f}(\boldsymbol x,t)+\boldsymbol g\big[\dot{\boldsymbol u}(t)+\dot{\boldsymbol d}(t)\big]\Big]\label{eq:sddot1}
\end{split}
\end{equation}
where $\dot{\boldsymbol f}(\boldsymbol x(t),t)\triangleq\frac{d}{dt}\boldsymbol f(\boldsymbol x(t),t)$ denotes the total derivative of the nominal dynamics along the system trajectory. The dots on $\boldsymbol u$ and $\boldsymbol d$ likewise denote their time derivatives. Since $\boldsymbol C$ and $\boldsymbol g$ are constant, their derivatives vanish.

Under Assumption \ref{assumption:ZOH}, $\dot{\boldsymbol u}(t)=\boldsymbol 0$ within each intersample interval. Hence, the control-input derivative term in \eqref{eq:sddot1} vanishes. Using first-order forward-difference approximations at the beginning of the interval, we obtain:
\begin{equation}
\begin{split}
&\ddot{\boldsymbol s}(t-h) =\boldsymbol{C}\dot{\boldsymbol f}(\boldsymbol x(t-h),t-h) +\boldsymbol{C}\boldsymbol{g}\dot{\boldsymbol d}(t-h)\\
&= \boldsymbol{C}\frac{\boldsymbol f(\boldsymbol x(t),t)-\boldsymbol f(\boldsymbol x(t-h),t-h)}{h}\\&\qquad+\boldsymbol{C}\boldsymbol{g}\frac{\boldsymbol d(t)-\boldsymbol d(t-h)}{h}+\mathcal{O}(h)
\label{eq:sddot2}
\end{split}
\end{equation}
Based on this derivation, we formalize the lumped local error term. In this paper, the System Time Delay Error characterizes the combined effect of the interval-wise variations of the nominal dynamics and lumped uncertainty over a single sampling interval. It is analytically defined as the corresponding variation terms scaled by the local interval length $h$.

\begin{definition}[\textit{System Time Delay Error}]\label{Def:SystemTDE}
The system time delay error $\boldsymbol \sigma(t,h)$ is defined as the sum of the disturbance variation and the nominal model variation over the interval $h$:
\begin{equation}
\begin{split}
\boldsymbol \sigma (t, h) \triangleq \underbrace{-\frac{h}{2}\boldsymbol{C}[\boldsymbol f(\boldsymbol x(t),t)-\boldsymbol f(\boldsymbol x(t-h),t-h)]}_{\text{Model TDE : }\boldsymbol{\sigma}^m}\\+\underbrace{\frac{h}{2}\boldsymbol{C}\boldsymbol{g}[\boldsymbol d(t)-\boldsymbol d(t-h)]}_{\text{Disturbance TDE : }\boldsymbol{\sigma}^d}
\end{split}\label{eq:TimeDelayError} 
\end{equation}
\end{definition}

For the variable-step implementation, let $t_n$ denote the $n$-th sampling instant and define the corresponding sampling interval $h_n=t_n-t_{n-1}$. As illustrated in Fig.~\ref{fig:DisctreteNotation}, the System TDE $\boldsymbol{\sigma}_n$ is evaluated over the realized interval $[t_{n-1},t_n]$ and used to determine the subsequent control interval $h_{n+1}$. For any signal $\boldsymbol{z}(t)$, we use the following discrete notation:
\begin{equation}
    \boldsymbol{z}_n\triangleq\boldsymbol{z}(t_n), \quad \Delta \boldsymbol{z}_n\triangleq\boldsymbol{z}_n-\boldsymbol{z}_{n-1}
    \label{eq:DiscreteNotation}
\end{equation}
The discrete System TDE in \eqref{eq:TimeDelayError} evaluated at the $n$-th sampling instant is then denoted by
\begin{equation}
    \boldsymbol{\sigma}_n \triangleq \boldsymbol{\sigma}(t_n, h_n)\label{eq:DiscreteTDE}
\end{equation}
Since there is a limit to how much the disturbances and the system can change during the time interval $h_n$, it is assumed that the Lipschitz continuity holds.
\begin{figure}
    \centering
    \includegraphics[width=1\linewidth]{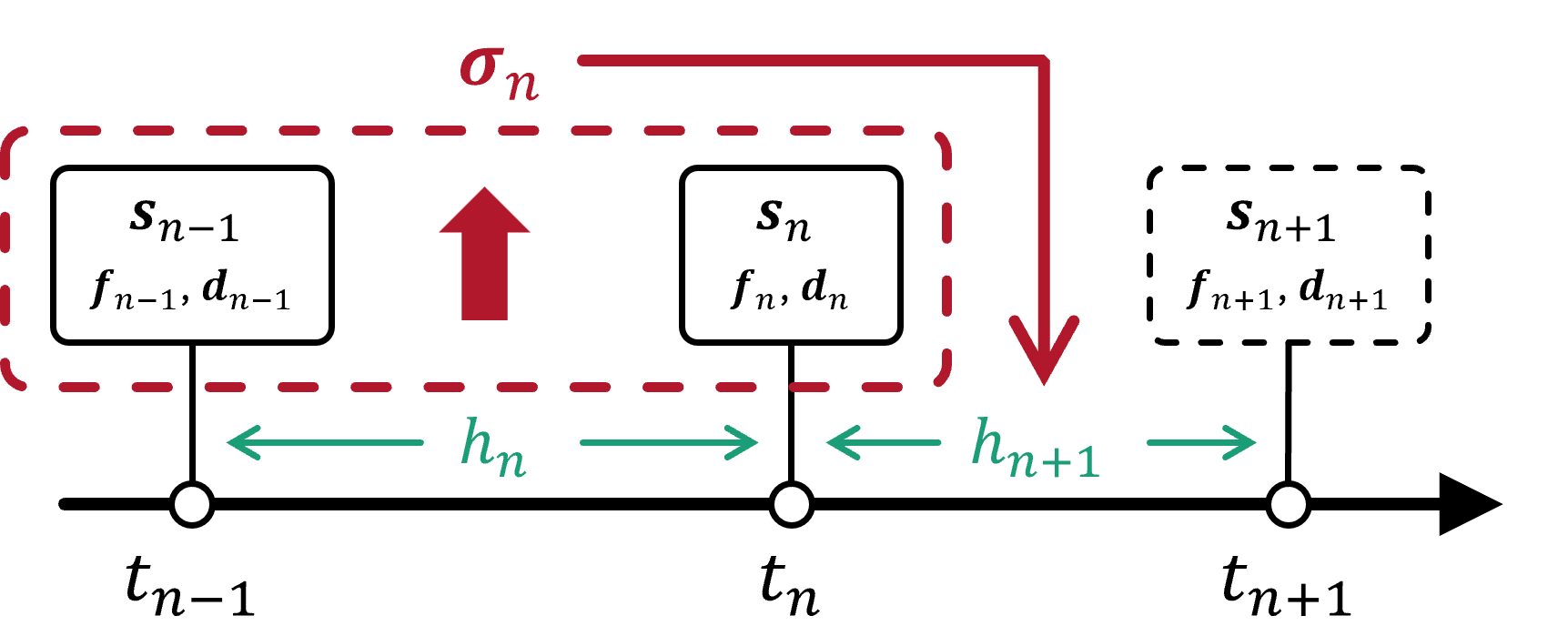}
    \caption{ Timing diagram of the variable-step implementation. 
The System TDE $\sigma_n$, evaluated from the variations between $t_{n-1}$ and $t_n$, is used to determine the subsequent control interval $h_{n+1}$.}
    \label{fig:DisctreteNotation}
\end{figure}
\begin{assumption}[\textit{Interval-wise Lipschitz Continuity}]
\label{assumption:Lipschitz}
For every realized sampling interval $[t_{n-1},t_n]$, the lumped uncertainty $\boldsymbol d(t)$ and the nominal dynamics $\boldsymbol f(\boldsymbol x(t),t)$ are assumed to satisfy an interval-wise Lipschitz condition~\cite{khalil2002nonlinear}. That is, there exist positive constants $L_d$ and $L_f$, independent of $n$, such that, for any $\tau_1,\tau_2\in[t_{n-1},t_n]$,
\begin{equation}
\begin{split}
\|\boldsymbol d(\tau_2)-\boldsymbol d(\tau_1)\|&\le L_d|\tau_2-\tau_1|,\\
\|\boldsymbol f(\boldsymbol x(\tau_2),\tau_2)-\boldsymbol f(\boldsymbol x(\tau_1),\tau_1)\|&\le L_f|\tau_2-\tau_1|.
\end{split}
\end{equation}
Consequently,
\begin{equation}
\|\Delta\boldsymbol d_n\|\le L_dh_n,\qquad
\|\Delta\boldsymbol f_n\|\le L_fh_n.
\end{equation}
In particular, $f_n \triangleq f(x_n,t_n)$. In addition, all derivatives required by the local Taylor and finite-difference expansions used below are assumed to exist and remain uniformly bounded over the admissible operating region, so that the constants implicit in the corresponding $\mathcal{O}(\cdot)$ terms are independent of the sampling index $n$.
\end{assumption}

Based on the definition of $\boldsymbol{\sigma}_n$ in \eqref{eq:TimeDelayError}--\eqref{eq:DiscreteTDE}, Assumption \ref{assumption:Lipschitz} provides an interval-wise upper bound on the magnitude of the TDE, $\|\boldsymbol{\sigma}_n\|$:

\begin{equation}
\begin{split}
\|\boldsymbol \sigma_n\| &= h_n\Big\|\frac{\boldsymbol{Cg}}{2}\Delta\boldsymbol d_n-\frac{\boldsymbol{C}}{2}\Delta\boldsymbol f_n\Big\|\\
&\leq h_n\Big[\frac{\|\boldsymbol{Cg}\|}{2}\ \|\Delta\boldsymbol d_n\|+\frac{\|\boldsymbol{C}\|}{2}\|\Delta\boldsymbol f_n\|\Big]\\
&\leq h_n^2 \Big(\frac{\|\boldsymbol{Cg}\|}{2}L_d+\frac{\|\boldsymbol{C}\|}{2}L_f\Big)\label{eq:TDEBound1}
\end{split}
\end{equation}

Therefore, we can derive the relationship between the error magnitude and the local sampling interval.
\begin{lemma}[\textit{Boundedness of Error}] \label{lemma:BoundnessOfError}
Under \eqref{eq:TDEBound1}, the magnitude of the System TDE over an arbitrary sampling interval of length $h_n$ is bounded by a constant multiple of $h_n^2$. Specifically,
\begin{equation}
\|\boldsymbol \sigma_n\|\leq\Big(\frac{\|\boldsymbol{Cg}\|}{2}L_d+\frac{\|\boldsymbol{C}\|}{2}L_f\Big)h_n^2=\Psi h_n^2 \label{eq:BoundnessOfError}
\end{equation}
Here, $h_n$ denotes the realized local sampling interval over $[t_{n-1},t_n]$, rather than a globally fixed sampling period.
\end{lemma}

With the System TDE and its interval-wise bound mathematically established, we return to the raw error dynamics. Specifically, successively substituting \eqref{eq:RawModel} and \eqref{eq:sddot2} into \eqref{eq:raw_ED}, and then applying the System TDE definition in \eqref{eq:TimeDelayError}, yields:
\begin{equation}
\begin{aligned}
\dot{\boldsymbol s}(t)
&\overset{\eqref{eq:RawModel}}{=}(\gamma-1)\frac{\boldsymbol s(t)}{h}
+\boldsymbol{Cg}\big[\boldsymbol d(t)-\boldsymbol d(t-h)\big] \\
&\qquad -\frac{1}{2}h\cdot\ddot{\boldsymbol s}(t-h)
+\mathcal{O}(h^2)\\
&\overset{\eqref{eq:sddot2}}{=}
(\gamma-1)\frac{\boldsymbol s(t)}{h}+\frac{1}{2}\boldsymbol{Cg} \big[\boldsymbol d(t)-\boldsymbol d(t-h)\big]\\
&\qquad
-\frac{1}{2}\boldsymbol C\big[\boldsymbol f(\boldsymbol x(t),t)-\boldsymbol f(\boldsymbol x(t-h),t-h)\big]
+\mathcal{O}(h^2)\\
&\overset{\eqref{eq:TimeDelayError}}{=}(\gamma-1)\frac{\boldsymbol s(t)}{h}+\frac{\boldsymbol\sigma(t,h)}{h}
+\mathcal{O}(h^2).
\end{aligned}
\label{eq:PiecewiseError}
\end{equation}
All remainder terms arising from these substitutions are of order $\mathcal{O}(h^2)$ and are therefore collected into a single $\mathcal{O}(h^2) term$. Evaluating \eqref{eq:PiecewiseError} at $t=t_n$ using the discrete notation in \eqref{eq:DiscreteNotation} and \eqref{eq:DiscreteTDE} yields the following local TDE-perturbed reaching dynamics.

\begin{lemma}[\textit{Local TDE-Perturbed Reaching Dynamics}]
\label{lemma:LocalTDEPerturbedDynamics}
Under the total control law in \eqref{eq:TotalControl} and Assumptions \ref{assumption:ZOH} and \ref{assumption:Lipschitz}, the sliding variable admits the following local dynamics up to the indicated higher-order residual over an interval of length $h_n$:

\begin{equation}
    \dot {\boldsymbol s}_n = (\gamma-1)\frac{\boldsymbol s_n}{h_n} + \frac{\boldsymbol \sigma_n}{h_n}+\mathcal{O}(h_n^2)
    \label{eq:DiscretePiecewiseError}
\end{equation}
\end{lemma}
Equation \eqref{eq:DiscretePiecewiseError} defines the local continuous-time error dynamics at $t_n$. To bridge this rate-level formulation with the discrete reaching objective in \eqref{eq:ControlObjective}, we predict the state at the next step $\boldsymbol{s}_{n+1}$ by projecting $\dot{\boldsymbol{s}}_n$ across the subsequent interval $h_{n+1}$. This forward one-step relation is approximated as:
\begin{equation}
    \boldsymbol{s}_{n+1} = \boldsymbol{s}_n + h_{n+1} \dot{\boldsymbol{s}}_n+\mathcal{O}(h_{n+1}^2)
    \label{eq:onestepForward}
\end{equation}

Substituting \eqref{eq:DiscretePiecewiseError} into \eqref{eq:onestepForward} introduces two higher-order residual terms associated with the local rate-level approximation and the subsequent state projection. Define the ratio between two consecutive sampling intervals as:
\begin{equation}
    r_n \triangleq \frac{h_{n+1}}{h_n}.
    \label{eq:StepRatio}
\end{equation}
Then, the discrete error dynamics can be written as
\begin{equation}
\begin{split}
\boldsymbol{s}_{n+1}& = \boldsymbol{s}_n+h_{n+1}\Big[ (\gamma-1)\frac{\boldsymbol{s}_n}{h_n}+\frac{\boldsymbol{\sigma}_n}{h_n}+\mathcal{O}(h_n^2) \Big]
+\mathcal{O}(h_{n+1}^2)\\
&=\big[1+r_n(\gamma-1)\big]\boldsymbol{s}_n+r_n\boldsymbol{\sigma}_n+\mathcal{O}(r_n h_n^3)+\mathcal{O}(r_n^2 h_n^2),
\label{eq:Expansion}
\end{split}
\end{equation}

\begin{remark}[\textit{Order Consistency under Bounded Step Ratios}] \label{re:OrderConsistency}
If we design an appropriate timestep update rule such that $0<r_n\leq\bar r<\infty$, where $\bar r$ denotes a finite uniform upper bound on the timestep ratio $r_n$, then the higher-order residual terms in \eqref{eq:Expansion} can be consistently expressed with respect to $h_n$. By the standard definition and algebra of Landau order symbols \cite{olver1997asymptotics}, the uniform boundedness of $r_n$ allows the factors $r_n$ and $r_n^2$ to be absorbed into the constants implicit in the $\mathcal{O}(\cdot)$ notation.
\begin{equation}
    \mathcal{O}(r_n h_n^3)=\mathcal{O}(h_n^3),
    \qquad
    \mathcal{O}(r_n^2 h_n^2)=\mathcal{O}(h_n^2).
\end{equation}
Furthermore, in the local small-step limit, $\mathcal{O}(h_n^3) \subseteq \mathcal{O}(h_n^2)$. Therefore, the two higher-order residual terms in
\eqref{eq:Expansion} can be collectively represented as
\begin{equation}
    \mathcal{O}(r_n h_n^3)+\mathcal{O}(r_n^2 h_n^2) = \mathcal{O}(h_n^2).
    \label{eq:CombinedRemainder}
\end{equation}
Accordingly, once the boundedness of $r_n$ is established by the proposed timestep update rule, the higher-order residual terms in \eqref{eq:Expansion} can be collectively represented as $\mathcal{O}(h_n^2)$.
\end{remark}

Applying Remark \ref{re:OrderConsistency} and defining the effective reaching parameter $\gamma_n^{\mathrm{eff}}\triangleq1+r_n(\gamma-1)$, \eqref{eq:Expansion} reduces to
\begin{equation}
\boldsymbol s_{n+1}=\gamma_n^{\mathrm{eff}}\boldsymbol s_n+r_n\boldsymbol\sigma_n+\underbrace{\mathcal O(h_n^2)}_{\boldsymbol\eta_n}.
\label{eq:effectiveErrorDynamics}
\end{equation}

Here, $\boldsymbol\eta_n$ denotes the aggregate higher-order residual. The effective reaching parameter $\gamma_n^{\mathrm{eff}}$ explicitly shows how the variable-step ratio $r_n$ modifies the nominal convergence rate $\gamma$ and governs the actual one-step attenuation of the sliding variable.

\begin{remark}[\textit{Boundedness of the Higher-Order Residual}] \label{re:ResidualBound}
Under Assumption~\ref{assumption:Lipschitz} and the stated bounded-step-ratio condition in Remark~\ref{re:OrderConsistency}, there exists a constant $\kappa>0$, independent of $n$, such that
\begin{equation}
\|\boldsymbol{\eta}_n\|\leq\kappa h_n^2.
\label{eq:EtaLocalBound}
\end{equation}
Therefore, if an appropriate timestep update rule provides a finite asymptotic upper bound
\begin{equation}
\limsup_{n\to\infty} h_n \leq h_{\max}<\infty,
\end{equation}
then the higher-order residual is also ultimately bounded as
\begin{equation}
\limsup_{n\to\infty}\|\boldsymbol{\eta}_n\|\leq\kappa h_{\max}^2.\label{eq:EtaUltimateBound}
\end{equation}
The existence of such a finite $h_{\max}$ is established later by the proposed timestep update rule.
\end{remark}


\subsection{Step Variation Algorithm}
As demonstrated in Section~\ref{sec:ErrorDynamics}, the error dynamics of the system driven by the control input \eqref{eq:TotalControl} are inherently governed by the System TDE, $\boldsymbol{\sigma}_n$. Accordingly, VTDC should shorten the control interval under rapid model or uncertainty variations and enlarge it under slowly varying conditions, providing adaptive control-update scheduling for spacecraft operation. Lemma \ref{lemma:BoundnessOfError} establishes that $\boldsymbol{\sigma}_n$ admits a quadratic upper bound with respect to the local interval length $h_n$. Building upon this quadratic relationship, we define the instantaneous TDE coefficient at the $n$-th step as follows:
\begin{equation}
\|\boldsymbol{\sigma}_n\|=\psi_n h_n^2\leq\Psi h_n^2, \qquad \psi_n \triangleq \frac{\|\boldsymbol{\sigma}_n\|}{h_n^2}
\label{eq:DiscreteBound1}
\end{equation}

\begin{assumption}[\textit{Persistent TDE Variation}]
\label{assumption:PersistentTDE}
For the uncertain systems considered in this study, the TDE-based scheduling signal is assumed to remain nonvanishing over the operating interval. Specifically, there exists a constant $\psi_{\min}>0$ such that
\begin{equation}
0<\psi_{\min}\leq\psi_n\leq\Psi,\qquad \forall n.
\end{equation}
\end{assumption}

Lemma~\ref{lemma:BoundnessOfError} establishes that the magnitude of the System TDE scales quadratically with the realized timestep. Motivated by adaptive timestep regulation in numerical integration, we seek a closed-form update law that regulates the TDE magnitude toward a prescribed tolerance while avoiding abrupt variations in the control interval. The resulting update rule is summarized in the following theorem.

\newtheorem{theorem}{Theorem}
\begin{theorem}[\textit{Timestep Update Rule}] \label{theo:UpdateRule} 
Under Assumptions~\ref{assumption:Lipschitz} and~\ref{assumption:PersistentTDE}, $\|\boldsymbol{\sigma}_n\|=\psi_n h_n^2>0$ for any $h_n>0$. For the $n$-th step $(n\geq1)$, let the next timestep be determined by the following update rule with constant design parameters $\mathrm{SF}\in(0,1)$ and $k>1$:
\begin{equation} 
h_{n+1} = \mathrm{SF}  \cdot \left( \frac{\mathrm{tol}}{\|\boldsymbol{\sigma}_n\|} \right)^{\frac{1}{2k}}\cdot h_n\label{eq:UpdateRule} 
\end{equation}
Equivalently, the step ratio is defined as: 
\begin{equation} r_{n} \triangleq \frac{h_{n+1}}{h_n} = \mathrm{SF} \left( \frac{\mathrm{tol}}{\|\boldsymbol{\sigma}_n\|} \right)^{\frac{1}{2k}} \label{eq:StepRatio2} 
\end{equation} 

Under the standard local-coefficient approximation adopted in adaptive timestep regulation, the nominal asymptotic TDE level is:

\begin{equation}\lim_{n \to \infty} \|\boldsymbol{\sigma}_n\| = \mathrm{SF}^{2k} \cdot \mathrm{tol} < \mathrm{tol}\end{equation}
\end{theorem}

\begin{proof}
The proposed stepsize adaptation is motivated by adaptive time-stepping techniques widely used in ordinary differential equation solvers, such as MATLAB ode45, and differential-algebraic equation solvers. In numerical integration methods, including Runge--Kutta methods, the local error estimate $E_n$ is commonly modeled as a  ($p$)-th power function of the timestep, $E_n = \phi_n h_n^p$ where $\phi_n$ is a local error coefficient. 

We adopt the standard elementary feedback stepsize-control principle from adaptive numerical integration \cite{butcher2016numerical}, \cite{soderlind2002automatic}, rather than deriving a new stepsize controller. The most basic update rule, namely the elementary update rule in numerical integration, is given by \cite{butcher2016numerical}:
\begin{equation}
    h_{n+1}=h_n \left(\frac{\mathrm{tol}}{E_n}\right)^{\frac{1}{p}} \label{eq:elementaryUpdate}
\end{equation}

Söderlind interpreted this structure not merely as a numerical heuristic, but as a feedback control problem, where the error estimate $E_n$ is regarded as the output, the timestep $h_n$ as the control input, and the prescribed tolerance $\mathrm{tol}$ as the setpoint \cite{soderlind2002automatic}. Taking the logarithm of \eqref{eq:elementaryUpdate} gives
\begin{equation}
    \ln{h_{n+1}} = \ln{h_n}+\frac{1}{p}\ln{\left(\frac{\mathrm{tol}}{E_n}\right)}
\end{equation}

This represents a negative feedback law for the error excess $E_n$ in the log-domain: when $E_n<\mathrm{tol}$, the timestep is increased, whereas when $E_n>\mathrm{tol}$, the timestep is decreased. 

Lemma~\ref{lemma:BoundnessOfError} provides exactly the error–stepsize structure required by the standard adaptive timestep controller, with the System TDE $\|\boldsymbol{\sigma}_n\|$ playing the role of the local error estimate $E_n$ and ($p=2$). Therefore, the simplest deadbeat update rule becomes:
\begin{equation}
    h_{n+1}=h_n \left(\frac{\mathrm{tol}}{\|\boldsymbol{\sigma}_n\|}\right)^{\frac{1}{2}} \label{eq:deadbeat}
\end{equation}
However, the elementary rule in \eqref{eq:elementaryUpdate} is a deadbeat-type update. Under the locally frozen error-coefficient approximation, it selects $h_{n+1}$ so that the predicted next-step TDE magnitude directly reaches $\mathrm{tol}$. Such a one-step correction can produce unnecessarily abrupt timestep variations. To moderate this correction, we introduce an intermediate target between the current TDE magnitude and the prescribed tolerance. Specifically, the next-step nominal target is defined by the weighted geometric mean
\begin{equation}
\|\boldsymbol{\sigma}_{n+1}\|_{\textrm{set}}=\left(\mathrm{tol}\cdot\|\boldsymbol{\sigma}_n\|^{k-1}\right)^{\frac{1}{k}} \label{eq:substitute}
\end{equation}
Here, $k>1$ determines the degree of smoothing. In the logarithmic domain. A smaller $k$ places greater weight on the prescribed tolerance $\mathrm{tol}$, whereas a larger $k$ retains more of the current TDE magnitude and therefore produces a more gradual correction. Taking the natural logarithm yields:

\begin{equation} 
\ln{\|\boldsymbol{\sigma}_{n+1}\|}_{\textrm{set}} = \frac{1}{k}\ln{(\mathrm{tol})} + \frac{k-1}{k}\ln{ \|\boldsymbol{\sigma}_{n}\|}. \label{eq:LogBound1} 
\end{equation}

Let $\lambda=\frac{k-1}{k}=1-\frac{1}{k}$. Since $k>1$, we have $\lambda\in(0,1)$. Under ideal setpoint tracking,  i.e., if the next-step TDE magnitude satisfies $\|\boldsymbol{\sigma}_{n+1}\|=\|\boldsymbol{\sigma}_{n+1}\|_{\mathrm{set}}$ at each update. the logarithmic TDE sequence follows the recurrence in \eqref{eq:LogBound1}. Note that the general solution of the first-order affine recurrence $a_{n+1} = xa_n + C$ (where $x \neq 1$) is $a_n = x^n a_0 + \frac{(1 - x^n)}{1 - x}C$. Based on this, we can obtain the general form of $\ln{\|\boldsymbol{\sigma}_{n}\|}$.

\begin{equation}
\ln {\|\boldsymbol{\sigma}_n\|} = \lambda^n \ln {\|\boldsymbol{\sigma}_0\|} + (1-\lambda^n)\cdot \ln{\big(\text{tol} \big)}\label{eq:LogBound2}
\end{equation}

Since $\lambda\in(0,1)$, $\lambda^n\to0$ as $n\to\infty$, and therefore 
\begin{equation}
\lim_{n \to \infty}\ln {\|\boldsymbol{\sigma}_n\|} =  \ln{(\text{tol})} \quad \to \quad
\lim_{n \to \infty}\|\boldsymbol{\sigma}_n\| =  \text{tol}\label{eq:SigmaTolConvergence}
\end{equation}

Applying the same target-to-timestep relation used in the deadbeat update \eqref{eq:deadbeat}, but replacing its target $\mathrm{tol}$ with the smoothed next-step target defined in \eqref{eq:substitute}, yields:
\begin{equation}
    h_{n+1}
    =h_n\cdot\Bigg[\frac{\overbrace{\left(\mathrm{tol}\cdot\|\boldsymbol{\sigma}_n\|^{k-1}\right)^{\frac{1}{k}}}^{\mathrm{tol}\to\eqref{eq:substitute}}}{\|\boldsymbol{\sigma}_n\|}\Bigg]^{\frac{1}{2}}
    =h_n\cdot\left( \frac{\mathrm{tol}}{\|\boldsymbol{\sigma}_n\|}\right)^{\frac{1}{2k}}
\label{eq:SlowUpdate} 
\end{equation} 

As is standard in numerical integration, the derivation of elementary update rules like \eqref{eq:elementaryUpdate} intrinsically relies on the assumption that the local error coefficient varies slowly (i.e., $\phi_{n+1} \approx \phi_n$). In the proposed VTDC framework, while this condition is not directly targeted, it naturally emerges as a practical byproduct of the timestep adaptation. As the algorithm actively bounds the system TDE magnitude, it inherently prevents the control intervals from becoming excessively large. As in standard adaptive timestep control, the elementary update is derived by locally freezing the error coefficient over consecutive steps. This approximation is adopted here only for the nominal timestep-regulation derivation. Subsequent coefficient variations are reflected in the next feedback update \cite{soderlind2002automatic}. Nevertheless, to provide a strictly conservative margin against unmodeled peak disturbances and to provide a standard conservative margin for dynamic transients, we adopt the standard practice in adaptive time-stepping of introducing a safety factor $\mathrm{SF} \in (0,1)$ as follows:
\begin{equation} h_{n+1} = \mathrm{SF} \cdot \left( \frac{\mathrm{tol}}{\|\boldsymbol{\sigma}_n\|} \right)^{\frac{1}{2k}}\cdot h_n, \quad \mathrm{SF}\in(0,1),\quad k>1 
\end{equation} 

Since $\mathrm{SF}$ is constant, the above update rule can be equivalently rewritten as follows:
\begin{equation}
\begin{split}
    h_{n+1} &=\left( \frac{\mathrm{SF}^{2k}\cdot\mathrm{tol}}{\|\boldsymbol{\sigma}_n\|}\right)^{\frac{1}{2k}}\cdot h_n \\
&=\Bigg(\frac{\left(\mathrm{SF}^{2k}\cdot\mathrm{tol}\cdot\|\boldsymbol{\sigma}_n\|^{k-1}\right)^{\frac{1}{k}}}{\|\boldsymbol{\sigma}_n\|}\Bigg)^{\frac{1}{2}}\cdot h_n 
\end{split}
\label{eq:NewConverge} 
\end{equation} 
Thus, \eqref{eq:SlowUpdate} can be interpreted as using the following next-step target $\|\boldsymbol{\sigma}_{n+1}\|_{\textrm{set}}$ instead of that defined in \eqref{eq:substitute}:
\begin{equation}
\|\boldsymbol{\sigma}_{n+1}\|_{\textrm{set}}=\left(\mathrm{SF}^{2k}\cdot\mathrm{tol}\cdot\|\boldsymbol{\sigma}_n\|^{k-1}\right)^{\frac{1}{k}}
\label{eq:ModerateGoal}
\end{equation} 
Applying the same derivation as in \eqref{eq:substitute}--\eqref{eq:SigmaTolConvergence} to the modified target in \eqref{eq:ModerateGoal} yields an asymptotic value of $\|\boldsymbol{\sigma}_n\|$ that is strictly smaller than $\mathrm{tol}$:
\begin{equation}
\lim_{n \to \infty}\|\boldsymbol{\sigma}_n\| =  \mathrm{SF}^{2k}\cdot\mathrm{tol}<\mathrm{tol}
\label{eq:LimNewConverge} 
\end{equation} 

Therefore, under the slowly varying TDE-coefficient condition $\psi_{n+1}\approx\psi_n$, the safety factor places the nominal asymptotic TDE level strictly below the prescribed tolerance.
\end{proof}

The proposed update rule in \eqref{eq:UpdateRule} not only regulates the System TDE but also the attenuation of the influence of the initial timestep. 

\begin{corollary}[\textit{Initial Timestep Independence}]\label{cor:initial_independence}
In the timestep update rule \eqref{eq:UpdateRule}, the influence of the initial heuristic timestep $h_0$ decays exponentially. As the number of iterations $n$ increases, the timestep sequence dependence on $h_0$ vanishes asymptotically.
\end{corollary}

\begin{proof}
The result follows from the contractive recurrence of $\ln h_n$
induced by the timestep update law in \eqref{eq:UpdateRule}. The complete derivation
is provided in Appendix~\ref{app:initial_independence}.
\end{proof}


\begin{corollary}[\textit{Boundedness of Timestep and Exclusion of Zeno Behavior}] \label{cor:NoZeno}
Under Assumption \ref{assumption:PersistentTDE}, the timestep sequence generated by \eqref{eq:UpdateRule} is ultimately bounded above and below. 
\begin{equation} 
\mathrm{SF}^k \sqrt{\frac{\mathrm{tol}}{\Psi}} \leq \liminf_{n\to\infty} h_n \leq \limsup_{n\to\infty} h_n \leq \mathrm{SF}^k \sqrt{\frac{\mathrm{tol}}{\psi_{\min}}} \label{eq:TimestepUltimateBound} \end{equation} 
Consequently, the timestep does not vanish asymptotically, and Zeno behavior is excluded. 
\end{corollary}
\begin{proof}
The result follows by applying the bounds
$\psi_{\min}\leq\psi_n\leq\Psi$ to the logarithmic recurrence of
the timestep sequence. The complete proof is provided in
Appendix~\ref{app:timestep_properties}.
\end{proof}

\begin{remark}[\textit{Exclusion of Zeno Behavior}] \label{remark:Zeno} 
From Corollary \ref{cor:NoZeno}, the asymptotic lower bound of the timestep is strictly positive: 
\begin{equation} 
h_{\min} \triangleq \mathrm{SF}^k \sqrt{\frac{\mathrm{tol}}{\Psi}} >0, 
\end{equation} 
because $\mathrm{SF}\in(0,1)$, $k>1$, $\mathrm{tol}>0$, and $\Psi<\infty$. Hence, the VTDC update interval remains bounded away from zero asymptotically. Zeno behavior refers to the occurrence of infinitely many discrete updates within a finite time interval\cite{johansson1999zeno}. Mathematically, as shown in \eqref{eq:Zeno}, Zeno behavior occurs when the sampling instants accumulate at a finite time $T^*$. 
\begin{equation}
\lim_{n \to \infty} t_n = T^* < \infty \label{eq:Zeno}
\end{equation}
Whether Zeno behavior occurs can be determined by the convergence of the following series \cite{johansson1999zeno}:
\begin{equation}
\begin{cases}\text{Occurrence of Zeno behavior:} & \sum_{n=0}^{\infty} h_n < \infty \\
\text{Avoidance of Zeno behavior:} & \sum_{n=0}^{\infty} h_n = \infty
\end{cases} 
\end{equation}

Since Corollary~\ref{cor:NoZeno} guarantees $\liminf_{n\to\infty}h_n\geq h_{\min}>0$, there exists a finite $N$ such that $h_n\geq h_{\min}/2$ for all $n\geq N$. Therefore,
\begin{equation}
\sum_{n=0}^{\infty}h_n \geq \sum_{n=N}^{\infty}\frac{h_{\min}}{2}=\infty.
\label{eq:ZenoSeries}
\end{equation}
Hence, the sampling instants cannot accumulate at a finite time, and Zeno behavior is excluded.

\end{remark}
\begin{corollary}[\textit{Bounds of the Step Ratio}] \label{cor:StepRatioBound}
Building upon the asymptotic timestep bounds $h_{\min}$ and $h_{\max}$ established in Corollary~\ref{cor:NoZeno}, asymptotic bounds on the consecutive step ratio $r_n$ are obtained. From
\eqref{eq:StepRatio2} and $\|\boldsymbol{\sigma}_n\|=\psi_n h_n^2$, the step ratio can be expressed as:
\begin{equation}
r_n = \mathrm{SF}\left(\frac{\mathrm{tol}}{\psi_n h_n^2}\right)^{\frac{1}{2k}}.
\end{equation}
Using the coefficient bounds $\psi_{\min}\leq\psi_n\leq\Psi$ together with the asymptotic timestep bounds $h_{\min}\leq h_n\leq h_{\max}$ yields: 
\begin{equation}
\mathrm{SF} \cdot \left(\frac{\mathrm{tol}} {\Psi h_{\max}^2}\right)^{\frac{1}{2k}}
\leq r_n \leq\mathrm{SF}\cdot \left(\frac{\mathrm{tol}}{\psi_{\min} h_{\min}^2}\right)^{\frac{1}{2k}}.
\end{equation}

Substituting the expressions for $h_{\min}$ and $h_{\max}$ from Corollary~\ref{cor:NoZeno} gives
\begin{equation}
r_{\min} = \left(\frac{\psi_{\min}}{\Psi}\right)^{\frac{1}{2k}},
\qquad r_{\max}=\left(\frac{\Psi}{\psi_{\min}}\right)^{\frac{1}{2k}},
\label{eq:StepRatioBounds}
\end{equation}
such that for all sufficiently large $n$
\begin{equation}
0<r_{\min}\leq r_n\leq r_{\max}<\infty.
\end{equation}
\end{corollary}
Therefore, the proposed timestep update rule satisfies both the bounded-step-ratio condition introduced in Remark~\ref{re:OrderConsistency} and the finite-timestep condition required in Remark~\ref{re:ResidualBound}.

\subsection{Convergence of Sliding Variable}
Corollaries~\ref{cor:NoZeno} and \ref{cor:StepRatioBound}
establish the timestep and step-ratio boundedness conditions invoked in Remarks~\ref{re:OrderConsistency} and \ref{re:ResidualBound}. Accordingly, the discrete sliding dynamics in \eqref{eq:effectiveErrorDynamics} can be represented as
\begin{equation}
\boldsymbol{s}_{n+1}=\gamma_n^{\mathrm{eff}}\boldsymbol{s}_n+
r_n\boldsymbol{\sigma}_n+\boldsymbol{\eta}_n,
\label{eq:StabilityDynamics}
\end{equation}
where
\begin{equation}
\limsup_{n\to\infty}\|\boldsymbol{\eta}_n\|\leq\kappa h_{\max}^2
\label{eq:ResidualAsymptoticBound}
\end{equation}

\begin{theorem}[\textit{Uniform Ultimate Boundedness of the Sliding Variable}] \label{Theo:UUB}
Under Assumption~\ref{assumption:Lipschitz}, \ref{assumption:PersistentTDE} and $|\gamma|<1$, a finite sensitivity parameter $k>1$ can always be selected such that the worst-case effective reaching magnitude satisfies $\bar{\gamma}<1$. Then, under the TDE-regulated regime characterized by Theorem \ref{theo:UpdateRule}, the sliding variable $\boldsymbol{s}_n$ under the proposed VTDC framework is uniformly ultimately bounded (UUB), with the following ultimate bound:
\begin{equation}
\limsup_{n\rightarrow\infty}{\|\boldsymbol{s}_n\|} \leq \frac{r_{\max}\cdot\mathrm{tol}+\kappa\cdot h_{\max}^2}{1-\bar{\gamma}} \label{eq:sBound}
\end{equation}
where $\bar{\gamma} \triangleq \max \left\{ |1+r_{\max}\cdot(\gamma-1)|, |1+r_{\min}\cdot(\gamma-1)| \right\}$.
\end{theorem}

\begin{proof} 
From \eqref{eq:StabilityDynamics}, define the aggregate perturbation as $\boldsymbol{p}_n\triangleq r_n\boldsymbol{\sigma}_n+\boldsymbol{\eta}_n$. Then, the discrete sliding dynamics in \eqref{eq:StabilityDynamics} can be written as:
\begin{equation}
\boldsymbol{s}_{n+1}=\gamma_n^{\mathrm{eff}}\boldsymbol{s}_n+\boldsymbol{p}_n
\label{eq:CompactStabilityDynamics}
\end{equation}
Using the ultimate bounds on the step ratio, System TDE, and higher-order residual, the aggregate perturbation satisfies:
\begin{equation}
\begin{split}
\|\boldsymbol{p}_n\|&
\leq r_n\|\boldsymbol{\sigma}_n\|+\|\boldsymbol{\eta}_n\|\\&
\leq r_{\max}\mathrm{tol}+\kappa h_{\max}^2\triangleq p_{\max}.
\end{split}
\label{eq:Bmax}
\end{equation}

Following the standard discrete-time Lyapunov analysis in \cite{furuta1990}, consider
\begin{equation}
V_n=\frac{1}{2}\boldsymbol{s}_n^T\boldsymbol{s}_n=\frac{1}{2}\|\boldsymbol{s}_n\|^2.
\label{eq:LyapunovFunction}
\end{equation}
Substituting \eqref{eq:CompactStabilityDynamics} into the one-step Lyapunov difference $\Delta V_n$ yields:
\begin{equation}
\begin{split}
\Delta V_n&= V_{n+1} - V_n = \frac{1}{2}\|\boldsymbol{s}_{n+1}\|^2 - \frac{1}{2}\|\boldsymbol{s}_n\|^2\\
&=\frac{1}{2}\left\|\gamma_n^{\mathrm{eff}}\boldsymbol{s}_n+\boldsymbol{p}_n\right\|^2-\frac{1}{2}\|\boldsymbol{s}_n\|^2\\
&=\frac{1}{2}\Big[\big((\gamma_n^{\mathrm{eff}})^2-1\big)\|\boldsymbol{s}_n\|^2+2\gamma_n^{\mathrm{eff}}
\boldsymbol{s}_n^T\boldsymbol{p}_n+\|\boldsymbol{p}_n\|^2\Big].
\end{split}
\label{eq:LyapunovDifference}
\end{equation}

Using the Cauchy--Schwarz inequality, $|\gamma_n^{\mathrm{eff}}\boldsymbol{s}_n^T\boldsymbol{p}_n|\leq|\gamma_n^{\mathrm{eff}}|\|\boldsymbol{s}_n\|\|\boldsymbol{p}_n\|$, we obtain
\begin{equation}
\Delta V_n\leq
\frac{1}{2}\underbrace{\left[-(1-|\gamma_n^{\mathrm{eff}}|^2)\|\boldsymbol{s}_n\|^2+
2|\gamma_n^{\mathrm{eff}}|p_{\max}\|\boldsymbol{s}_n\|+p_{\max}^2\right]}_{\Phi(\|\boldsymbol{s}_n\|)}.
\label{eq:LyapunovUpperBound}
\end{equation}

For $1-|\gamma_n^{\mathrm{eff}}|^2 > 0$ (which is equivalent to $|\gamma_n^{\mathrm{eff}}| < 1$), the function $\Phi(\|\boldsymbol{s}_n\|)$ is a concave quadratic function. Setting $\Phi(\|\boldsymbol{s}_n\|) = 0$, the positive root is:

\begin{equation}
\|\boldsymbol{s}_n\| =\frac{p_{\max}(|\gamma_n^{\mathrm{eff}}|+1)}{1-|\gamma_n^{\mathrm{eff}}|^2}=\frac{p_{\max}}{1-|\gamma_n^{\mathrm{eff}}|}.\label{eq:BoundSol}
\end{equation}

Therefore, whenever $|\gamma_n^{\mathrm{eff}}|<1$, $\Delta V_n<0$ for all $\|\boldsymbol{s}_n\|$ greater than the bound in \eqref{eq:BoundSol}. Consequently, the sliding variable is driven toward an ultimate bounded region whose size depends on the worst-case magnitude of the effective reaching parameter. Recall from Section~\ref{sec:ErrorDynamics} that $\gamma_n^{\mathrm{eff}}=1+r_n(\gamma-1)$.

By Corollary~\ref{cor:StepRatioBound}, the step ratio satisfies $r_{\min}\leq r_n\leq r_{\max}$, while $\gamma$ is a prescribed reaching parameter satisfying $|\gamma|<1$. Since $|1+r(\gamma-1)|$ is a convex function of $r$, its maximum over the bounded interval $[r_{\min},r_{\max}]$ occurs at one of the two endpoints. We therefore define the worst-case effective reaching magnitude as:
\begin{equation}
\bar{\gamma}\triangleq
\max\left\{\left|1+r_{\min}(\gamma-1)\right|,\left|1+r_{\max}(\gamma-1)\right|\right\}\label{eq:GammaBar}
\end{equation}

If $\bar{\gamma}<1$, then $|\gamma_n^{\mathrm{eff}}|\leq\bar{\gamma}<1$ for every admissible step ratio. Therefore,
\begin{equation}
\limsup_{n\rightarrow\infty}\|\boldsymbol{s}_n\|
\leq\frac{p_{\max}}{1-\bar{\gamma}}
=\frac{r_{\max}\mathrm{tol}+\kappa h_{\max}^2}{1-\bar{\gamma}}
\label{eq:UUBResult}
\end{equation}
Hence, the sliding variable is uniformly ultimately bounded provided that $\bar{\gamma}<1$.

To guarantee $\bar{\gamma}<1$, it is sufficient to require:
\begin{equation}
    -1<1+r_n(\gamma-1)<1
\end{equation}
for all admissible $r_n$, since $r_n>0$ and $\gamma <1$, the upper inequality is automatically satisfied. The lower inequality is most restrictive at $r_n = r_{\max}$, because $1 + r_{n}(\gamma - 1)$ decreases monotonically with $r_n$. Hence, it is sufficient to require $1 + r_{\max}(\gamma - 1) > -1$, which rearranges as follows with Corollary~\ref{cor:StepRatioBound}:
\begin{equation}
r_{\max}=\left( \frac{\Psi}{\psi_{\min}} \right)^{\frac{1}{2k}} < \frac{2}{1-\gamma}
\label{eq:gammaCondition}
\end{equation}
Taking the natural logarithm of both sides isolates the sensitivity parameter $k$, yielding the conditional stability criterion:
\begin{equation}
k > \frac{\ln(\Psi/\psi_{\min})}{2\ln(2/(1-\gamma))}
\label{eq:kCondition}
\end{equation}
Since $0<\psi_{\min}\leq\Psi<\infty$ and $|\gamma|<1$, the right-hand side of \eqref{eq:kCondition} is finite. Hence, a finite $k>1$ satisfying this condition always exists. Therefore, choosing $k$ according to \eqref{eq:kCondition} limits the worst-case step ratio sufficiently to ensure $\bar{\gamma} < 1$, thereby establishing the UUB result in \eqref{eq:sBound}.
\end{proof}

\textbf{Practical Feasibility:} Although Theorem \ref{Theo:UUB} establishes conditional stability, the requirement in \eqref{eq:kCondition} is mild. Rearranging the condition yields a maximum tolerable variation ratio of $\Psi/\psi_{\min} < \big( 2/(1-\gamma) \big)^{2k}$. With the parameters used in our simulations ($\gamma = 0.5$, $k = 24$), this allowable ratio reaches an immense margin of $\approx 7.92 \times 10^{28}$. Thus, for the parameter values considered in this study, the sufficient stability condition admits a very large range of TDE-coefficient variation.

\begin{corollary}[\textit{Steady-State Convergence Radius}] \label{cor:SteadyStateBound} 
Suppose that the timestep adaptation reaches a steady-state regime in which $h_n\to h_{ss}>0$. Then $r_n =h_{n+1}/h_n \to 1$, and consequently $\gamma^{\mathrm{eff}}_n \to \gamma$.
 Accordingly, the discrete sliding dynamics in \eqref{eq:effectiveErrorDynamics} asymptotically reduce to
\begin{equation}
\boldsymbol{s}_{n+1}=\gamma\boldsymbol{s}_n+\boldsymbol{\sigma}_n+\mathcal{O}(h_{\mathrm{ss}}^2).
\label{eq:SteadyStateDynamics}
\end{equation}
Here, $\mathcal{O}(h_{\mathrm{ss}}^2)$ denotes the second-order higher-order residual in the local small-step sense as $h_n\to h_{ss}$. Using the asymptotic TDE regulation $\|\boldsymbol{\sigma}_n\|\leq\mathrm{tol}$ and boundedness of residual  $\mathcal{O}(h_{\mathrm{ss}}^2)\leq\kappa\cdot h_{\mathrm{ss}}^2$, the corresponding steady-state ultimate bound becomes:
\begin{equation}
\limsup_{n\to\infty}\|\boldsymbol{s}_n\|
\leq\frac{\mathrm{tol}}{1-|\gamma|}+\frac{\kappa\cdot h_{\mathrm{ss}}^2}{1-|\gamma|}.
\label{eq:SteadyStateBound}
\end{equation}
Thus, once the timestep variation settles, the steady-state sliding-variable bound approaches the nominal fixed-step result to leading order, with only a second-order correction associated with the finite control interval.
\end{corollary}
\section{Practical Implementation}\label{sec:PracticalImple}
\subsection{VTDC Implementation Procedure}
\begin{figure}
  \includegraphics[width=1\linewidth]{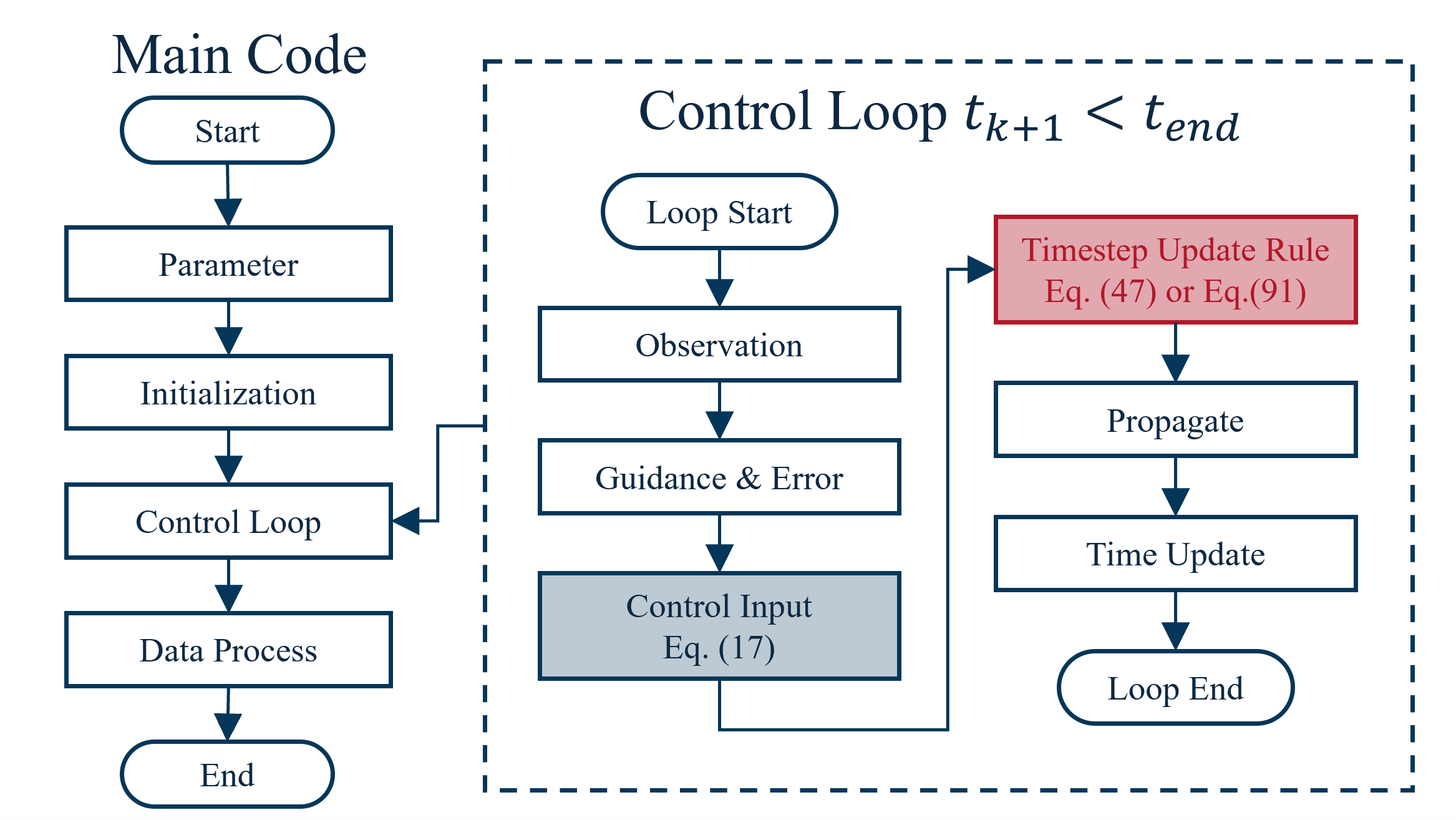}
  \caption{Computational flow of the proposed VTDC implementation, including control-input computation and adaptive determination of the subsequent control-update interval}
  \label{fig:VTDCAlgorithm}
\end{figure}
The proposed VTDC follows the onboard control-update sequence illustrated in Fig.~\ref{fig:VTDCAlgorithm}, which is implemented in MATLAB here for numerical validation. At each control update instant $t_n$, the current state is first observed, and the guidance command and corresponding tracking error are evaluated. The control input is then computed from \eqref{eq:TotalControl} and held constant over the subsequent control interval according to the zero-order-hold implementation. The next control interval is determined algebraically by the timestep update rule in the practical form introduced in \eqref{eq:PracticalUpdateRule}. The system dynamics are subsequently propagated over the determined interval, after which the simulation time is advanced to the next control update instant $t_{n+1}$. This sequence is repeated until the terminal simulation time is reached.

It is important to distinguish the control-update interval from the numerical integration step used to propagate the continuous-time spacecraft dynamics. VTDC adaptively varies only the former, whereas the integration step should remain sufficiently small and fixed so that numerical stepsize adaptation does not interfere with the intended control scheduling behavior. Accordingly, adaptive-step numerical integrators such as ode45 were not used in the simulations. Instead, the spacecraft dynamics were propagated using a fixed-step fifth-order Runge--Kutta method with an integration step of $10^{-4}$~s, while the control-update interval was varied independently according to \eqref{eq:UpdateRule}.

\subsection{Practical System-TDE Estimation}
In the System TDE defined in \eqref{eq:TimeDelayError}, the Model TDE can be evaluated from the known nominal dynamics and measured states. In contrast, the Disturbance TDE depends on the variation of the lumped uncertainty $\boldsymbol{d}(t)$, which is not directly available without additional estimation. To avoid introducing a separate disturbance observer, we estimate this variation from the robust control-input variation. The objective is not exact disturbance reconstruction, but a computationally inexpensive scheduling surrogate constructed from controller-available quantities. For this scheme, we assume that the actuator-generated control variation can compensate for the interval-wise variation of the nominal dynamics and disturbances.

It is not $\boldsymbol\sigma$ itself but its magnitude, $\|\boldsymbol\sigma\|$, that is used to determine the time step. In this case, the following holds for $\|\boldsymbol\sigma\|$.
\begin{equation}
    \|\boldsymbol\sigma_n\|=\|\boldsymbol\sigma^m_n+\boldsymbol\sigma^d_n\|\leq\|\boldsymbol\sigma^m_n\|+\|\boldsymbol\sigma^d_n\| \label{eq:TDEtri}
\end{equation}

Equation~\eqref{eq:TDEtri} provides the key principle for constructing the practical scheduling surrogate. By separately evaluating the magnitudes of the model and disturbance TDE components and adding them, possible cancellation between the two vector components is deliberately neglected. Therefore, for given component estimates, this additive construction biases the estimated TDE magnitude toward a larger value. Since the timestep update law is monotonically decreasing with respect to the TDE magnitude, such overestimation results in a shorter, and hence more conservative, subsequent control interval.

\subsubsection{Estimation of Model TDE}
The equivalent control $\boldsymbol{u}^{eq}_n$ is defined as $\boldsymbol{u}^{eq}_n = - (\boldsymbol{Cg})^{-1}\boldsymbol{C}\boldsymbol{f}_n$ as shown in \eqref{eq:EquivalentControl}, so $\boldsymbol{C}\boldsymbol{f}_n$ can be expressed in terms of $\boldsymbol{u}^{eq}_n$ as follows.

\begin{equation}
\boldsymbol{C}\boldsymbol{f}_n=-(\boldsymbol{Cg})\boldsymbol u^{eq}_n
\end{equation}

Based on this, the model TDE can be approximated by $\boldsymbol{u}^{eq}(t)$ as follows:
\begin{equation}
\begin{split}
\|\boldsymbol\sigma^m_n\|&=\frac{h_n}{2}\|\boldsymbol{C}\Delta\boldsymbol{f}_n\|=\frac{h_n}{2}\Big\|\boldsymbol{Cg}\Delta\boldsymbol u^{eq}_n\Big\|\\
&\leq \frac{\|\boldsymbol{Cg}\|}{2}h_n\|\Delta\boldsymbol u^{eq}_n\|=\|\hat{\boldsymbol\sigma}^m_n\|
\end{split}\label{eq:ModelTDEEstimation}
\end{equation}

\subsubsection{Practical Surrogate for Disturbance TDE}

From the reduced dynamics in \eqref{eq:ReducedDynamics}, the
interval-wise variations satisfy
\begin{equation}
\Delta\dot{\boldsymbol{s}}_n =\boldsymbol{Cg}\left(\Delta\boldsymbol{u}^d_n+\Delta\boldsymbol{d}_n\right)\quad\text{or}\quad
\boldsymbol{Cg}\Delta\boldsymbol{d}_n=\Delta\dot{\boldsymbol{s}}_n-\boldsymbol{Cg}\Delta\boldsymbol{u}^d_n.
\end{equation}
Although \(\Delta\boldsymbol{u}^d_n\) does not directly represent \(\Delta\boldsymbol{d}_n\) when $\Delta\dot{\boldsymbol{s}}_n\neq\boldsymbol{0}$, the robust control input $\boldsymbol{u}^d$ is designed to compensate for the lumped uncertainty while driving the sliding dynamics toward the sliding manifold. Its interval-wise variation therefore reflects the overall local compensation demand, including the effect of disturbance variation. As the sliding dynamics approach the quasi-sliding regime and $\Delta\dot{\boldsymbol{s}}_n$ becomes small, the contribution of the transient sliding dynamics diminishes, and $\Delta\boldsymbol{u}^d_n$ more directly reflects $\Delta\boldsymbol{d}_n$. Accordingly, to avoid direct disturbance measurement or additional derivative estimation, we adopt the following computationally inexpensive scheduling surrogate:
\begin{equation}
\|\hat{\boldsymbol{\sigma}}^d_n\| \triangleq \frac{h_n}{2}\|\boldsymbol{Cg}\|
\|\Delta\boldsymbol{u}^d_n\|.
\label{eq:DisturbanceTDEEstimation}
\end{equation}

\subsubsection{Practical System-TDE Surrogate}\label{cor:SystemTDE}
For onboard implementation, the individual System-TDE components are replaced by quantities available from the controller. Following the additive construction in \eqref{eq:TDEtri}, the practical System-TDE surrogate is defined by combining the equivalent- and robust-control variations as follows:
\begin{equation}
    \|\boldsymbol{\hat{\sigma}}_n\| =\|\boldsymbol{\hat{\sigma}}^m_n\|+\|\boldsymbol{\hat{\sigma}}^d_n\| = \frac{\|\boldsymbol{Cg}\|}{2} h_n\left( \|\Delta \boldsymbol{u}^{eq}_n\| + \|\Delta \boldsymbol{u}^d_n\| \right) \label{eq:SystemTDEEstimation}
\end{equation}

The construction in \eqref{eq:SystemTDEEstimation} intentionally favors conservative scheduling in two respects. First, the induced-norm bound in \eqref{eq:ModelTDEEstimation} provides an upper bound for the model-related contribution. Second, following \eqref{eq:TDEtri}, the two component magnitudes are added separately, thereby neglecting any cancellation between the model and disturbance TDE components. Since the timestep update law decreases monotonically with the TDE magnitude, a larger surrogate produces a shorter subsequent control interval. Although the disturbance-related term does not exactly coincide with the disturbance variation during transient sliding motion, these conservative components provide practical margin against the resulting discrepancy. The resulting conservatism is assessed numerically in Fig.~\ref{fig:TDEValidation} of Section~\ref{sec:NumDef}, where $\|\boldsymbol{\hat{\sigma}}_n\|$ remains above the actual System TDE for the considered maneuver. Using this estimated System TDE, the practical timestep update rule is rewritten as follows:

\begin{equation}
h_{n+1}=\mathrm{SF}\cdot\left(\frac{\mathrm{tol}}{\|\boldsymbol{\hat{\sigma}}_n\|+\epsilon}\right)^{\frac{1}{2k}}\cdot h_n,
\label{eq:PracticalUpdateRule}
\end{equation}
where $\epsilon\geq0$ is an optional safeguard constant that prevents excessive timestep growth when $\|\boldsymbol{\hat{\sigma}}_n\|$ approaches zero under nearly stationary operating conditions. For the simulations in this study, $\epsilon=0$ was used because the scheduling surrogate remained sufficiently nonzero throughout the considered maneuver.

\section{Numerical Demonstration} \label{sec:NumDef}
This section evaluates the VTDC  framework proposed in Section~\ref{sec:framework} and \ref{sec:PracticalImple} through a nonlinear three-axis spacecraft attitude maneuver under representative model, disturbance, and actuator uncertainties. Numerical simulations are used to evaluate the closed-loop tracking performance, practical TDE estimation, timestep-dependent error scaling, adaptive timestep behavior, and consistency with the derived timestep and stability bounds.

\subsection{Problem Statement: Spacecraft Attitude Control}
To evaluate the performance of the proposed VTDC, we consider the attitude control of a rigid-body spacecraft. The attitude error is described by the error quaternion $\boldsymbol{q}^e = [q_0^e, (\boldsymbol{q}_v^e)^T]^T \in \mathbb{R}^4$, defined as the quaternion product between the desired orientation $\boldsymbol{q}^d$ and the body orientation $\boldsymbol{q}^b$. The attitude-error quaternion is defined using the standard
quaternion composition convention~\cite{markley2014fundamentals}:
\begin{equation}
\boldsymbol{q}^e = (\boldsymbol{q}^d)^{*} \otimes \boldsymbol{q}^b\label{eq:err_quat}
\end{equation}
In \eqref{eq:err_quat}, the operator ``$\otimes$'' denotes the quaternion multiplication. For any two quaternions $\boldsymbol{p}=[p_0, \boldsymbol p_v^T]^T$ and $\boldsymbol{q}=[q_0, \boldsymbol q_v^T]^T$, their product is evaluated as $\boldsymbol{p} \otimes \boldsymbol{q} = \big[ p_0 q_0 - \boldsymbol{p}_v \cdot \boldsymbol{q}_v, \ (p_0 \boldsymbol{q}_v + q_0 \boldsymbol{p}_v + \boldsymbol{p}_v \times \boldsymbol{q}_v)^T \big]^T$. The superscript ``$(\cdot)^*$'' denotes the quaternion conjugate, defined as $\boldsymbol{q}^* = [q_0, -\boldsymbol{q}_v^T]^T$. 

Furthermore, because unit quaternions double-cover $SO(3)$, the error quaternions $\boldsymbol q^e$ and $-\boldsymbol q^e$ represent the same physical attitude. To avoid the unwinding phenomenon\cite{bhat2000unwinding}, we select the equivalent quaternion representation with a nonnegative scalar component, consistent with the geometric quaternion formulation in \cite{lopez2021sliding}:
\begin{equation}
\boldsymbol{q}^e_{\text{mod}} =\begin{cases}\boldsymbol{q}^e, & \text{if } q_0^e \ge 0 \\-\boldsymbol{q}^e, & \text{if } q_0^e < 0\end{cases}\label{eq:Unwinding}
\end{equation}
This representation confines the principal attitude error to $[0,\pi]$ and selects a shortest rotational path. Hereafter, $\boldsymbol{q}^e$ denotes the sign-adjusted representation in \eqref{eq:Unwinding}.

The error angular velocity $\boldsymbol{\omega}^e \in \mathbb{R}^3$ represents the body angular velocity relative to the desired angular velocity, with both terms expressed in the body-fixed frame.
\begin{equation}
\boldsymbol{\omega}^e = \boldsymbol{\omega}^b - \boldsymbol{R}\boldsymbol{\omega}^d\label{eq:ErrorAngVel}
\end{equation}
where $\boldsymbol{\omega}^b$ denotes the angular velocity of the spacecraft body frame and $\boldsymbol{\omega}^d$ is the desired angular velocity relative to the inertial frame. Here, the rotation matrix $\boldsymbol{R}$ serves to transform the desired angular velocity into the body frame, ensuring that the error dynamics are consistently formulated within the body-fixed coordinate system.

Based on the error states defined above, we establish the state vector $\boldsymbol{x} \in \mathbb{R}^6$ in \eqref{eq:StateVector} to describe the general behavior of the system as follows:
\begin{equation}\boldsymbol{x} = \begin{bmatrix} \boldsymbol{q}_v^e \\ \boldsymbol{\omega}^e \end{bmatrix} \in\mathbb{R}^6
\end{equation}

The derivative of this state vector can be defined in terms of the standard quaternion kinematics and rigid-body rotational dynamics of the spacecraft \cite{markley2014fundamentals}:

\begin{equation}
\begin{split}
    \dot{\boldsymbol{q}}_v^e &= \frac{1}{2}(q_0^e \boldsymbol{I} + (\boldsymbol{q}_v^e)^\times)\boldsymbol{\omega}^e\\
    \boldsymbol{J}\dot{\boldsymbol{\omega}}^b &=\boldsymbol{J}(\dot{\boldsymbol{\omega}}^e+\boldsymbol{\omega}^e \times \boldsymbol{R}\boldsymbol{\omega}^d - \boldsymbol{R}\dot{\boldsymbol{\omega}}^d)\\
    &=-\boldsymbol{\omega}^b \times \boldsymbol{J}\boldsymbol{\omega}^b +\boldsymbol{u} + \boldsymbol{d}
\end{split}
\end{equation}

In this formulation, the operator $(\cdot)^\times$ denotes the skew-symmetric matrix used to represent the vector cross product as a matrix multiplication, and $\boldsymbol{J}$ denotes the inertia matrix of the spacecraft. By combining these kinematic and rotational dynamic equations, we can define the nominal dynamics $\boldsymbol{f}(\boldsymbol{x},t)$ of the general reaching dynamics in \eqref{eq:ReachingDynamics} as follows:

\begin{equation}
\boldsymbol{f}(\boldsymbol{x}, t) = \begin{bmatrix} \frac{1}{2}(q_0^e \boldsymbol{I} + (\boldsymbol{q}_v^e)^\times)\boldsymbol{\omega}^e \\ \boldsymbol{J}^{-1}(-\boldsymbol{\omega}^b \times \boldsymbol{J}\boldsymbol{\omega}^b) - (\boldsymbol{\omega}^e \times \boldsymbol{R}\boldsymbol{\omega}^d - \boldsymbol{R}\dot{\boldsymbol{\omega}}^d) \end{bmatrix}
\end{equation}

The nominal dynamics $\boldsymbol{f}(\boldsymbol{x}, t)$ encapsulates all known nonlinearities of the system. The VTDC then utilizes this nominal model in conjunction with delayed data to effectively cancel the lumped uncertainty $\boldsymbol{d}$, ensuring robust tracking performance.

To implement the robust control framework in Section~\ref{sec:framework}, the sliding variable $\boldsymbol{s} \in \mathbb{R}^3$ is defined to characterize the desired error convergence behavior. Following the quaternion-based sliding-variable construction of Lopez and Slotine \cite{lopez2021sliding}, we consider:
\begin{equation}
\boldsymbol{s} = \boldsymbol{\omega}^e + \boldsymbol{\beta}\mathrm{sgn}_+(q_0^e) \boldsymbol{q}_v^e
\end{equation}
where $\boldsymbol{\beta} = \textrm{diag}[\beta_1, \beta_2, \beta_3] \succ 0$ is a positive definite diagonal gain matrix. As discussed in the previous section, we select the equivalent quaternion representative with $q_0^e \ge 0$. Consequently, the term $\mathrm{sgn}_+(q_0^e)$ can be simplified to unity without loss of generality. To align this definition with the state-space formulation $\boldsymbol{s} = \boldsymbol{Cx}$ in \eqref{eq:SlidingVariable}, the sliding variable is expressed using the sliding manifold gain matrix $\boldsymbol{C} \in \mathbb{R}^{3\times6}$ as follows:

\begin{equation}
\boldsymbol{s} =  \boldsymbol{\omega}^e + \boldsymbol{\beta}\boldsymbol{q}_v^e = \underbrace{\begin{bmatrix} \boldsymbol{\beta} & \boldsymbol{I}_3 \end{bmatrix}}_{\boldsymbol{C}} \underbrace{\begin{bmatrix} \boldsymbol{q}_v^e \\ \boldsymbol{\omega}^e \end{bmatrix}}_{\boldsymbol{x}} \label{eq:sliding_def}
\end{equation}
Unlike the canonical representation in Section~\ref{sec:framework}, $\boldsymbol{\omega}^e$  is not the time derivative of $\boldsymbol{\beta}\boldsymbol{q}_v^e$; however, under the manifold in \eqref{eq:sliding_def}, $\boldsymbol{s}=\boldsymbol{0}$  yields the quaternion-based stable reduced dynamics established in \cite{lopez2021sliding}. The following uncertainties are present in this attitude control simulation.

\begin{enumerate}[label=\textcircled{\small\arabic*}]
    \item \textbf{Moment of Inertia Uncertainty:} 
    The spacecraft's moment of inertia is composed of a nominal term $\boldsymbol{J}_0$ and an uncertainty $\Delta\boldsymbol{J}$, expressed as $\boldsymbol{J}=\boldsymbol{J}_0+\Delta\boldsymbol{J}\in\mathbb{R}^{3\times3}$. This yields uncertainty $\boldsymbol{d}_{\Delta \boldsymbol{J}}$.
    
    \item \textbf{Disturbance Torque:} Disturbance torque $\boldsymbol{\tau}_{d}$ is the sum of three types of disturbance torques acting on the spacecraft during the simulation: Gravity Gradient Torque ($\boldsymbol{\tau}_{gg}$)  \cite{markley2014fundamentals}, Constant Bias Torque  ($\boldsymbol{\tau}_{bias}$), and Sine Torque ($\boldsymbol{\tau}_{sin}$).
    \begin{align*}
    \boldsymbol{\tau}_{gg} &= 3 \frac{\mu}{\|\boldsymbol{r}_I\|^3} (\hat{\boldsymbol{r}}_b \times (\boldsymbol{J} \hat{\boldsymbol{r}}_b))\\ 
    \boldsymbol{\tau}_{bias} &= 0.01 \times \begin{bmatrix}0.5& -0.3& 0.2\end{bmatrix}^T\\    
    \boldsymbol{\tau}_{sin} &= 0.01 \times \begin{bmatrix} \sin 0.1t & \sin 0.2t & \sin 0.3t \end{bmatrix}^T
    \end{align*}
    where $\mu$ is the Earth's standard gravitational parameter, and $\|\boldsymbol{r}_I\|$ is the geocentric distance of the spacecraft. The unit vector $\hat{\boldsymbol{r}}_b$ represents the direction from the center of the Earth to the spacecraft center of mass, transformed into the body-fixed coordinate system. 
    \item \textbf{Actuator Failure: }In the simulation, actuator performance is degraded such that only approximately 90\% of the original control input is applied.
    $$ \boldsymbol{J} \dot{\boldsymbol{\omega}}^b = -\boldsymbol{\omega}^b \times \boldsymbol{J} \boldsymbol{\omega}^b + \underbrace{(\boldsymbol{I} - \boldsymbol{a}) \boldsymbol{u}}_{\text{Actuator fail}} + \boldsymbol{d} $$
    Here, the Actuator Failure Parameter $\boldsymbol{a}$ is defined as follows:
    $$ \boldsymbol{a} = 0.1 \times \left( \boldsymbol{I} + 0.1 \cdot \text{diag} \begin{bmatrix} \sin(\frac{t}{10}) & \cos(\frac{t}{15}) & \sin(\frac{t}{20}) \end{bmatrix} \right) $$
\end{enumerate}
By mapping the attitude dynamics and uncertainties into the general reaching dynamics form \eqref{eq:ReachingDynamics}, we obtain the following formulated expression:
\begin{equation}
\begin{split}
\dot{\boldsymbol s}(t) &= \underbrace{\begin{bmatrix} \boldsymbol{\beta} & \boldsymbol{I}_3 \end{bmatrix}}_{\boldsymbol C} \Bigg(\underbrace{\begin{bmatrix} \boldsymbol{0}_{3\times3} \\ \boldsymbol{J}_0^{-1} \end{bmatrix}}_{\boldsymbol g}\Big[\boldsymbol u(t)+\underbrace{\boldsymbol{\tau}_d+\boldsymbol{d}_{\Delta \boldsymbol{J}}-\boldsymbol{a} \boldsymbol{u}}_{\boldsymbol d(t)}\Big]\\
&+\underbrace{\begin{bmatrix} \frac{1}{2}(q_0^e \boldsymbol{I} + (\boldsymbol{q}_v^e)^\times)\boldsymbol{\omega}^e \\ \boldsymbol{J}_0^{-1}\big(-\boldsymbol{\omega}^b\times \boldsymbol{J}_0\boldsymbol{\omega}^b\big) - \big(\boldsymbol{\omega}^e\times \boldsymbol{R}\boldsymbol{\omega}^d - \boldsymbol{R}\dot{\boldsymbol{\omega}}^d\big) \end{bmatrix}}_{\boldsymbol f(\boldsymbol x,t)}\Bigg) \label{eq:ReachingDynamicsAttitude}
\end{split}
\end{equation}

To prevent abrupt initial transients and excessive control effort, a smooth rest-to-rest slew trajectory is generated using a fifth-order polynomial time-scaling profile $\rho(\tau)$ for a specified slewing time $t_s$, where $\tau=t/t_s$ denotes the normalized time. The quintic profile imposes zero angular velocity and acceleration at both endpoints, providing a smooth spacecraft slew trajectory \cite{caubet2013optimal}. This profile is determined based on the total principal angle $\Delta \theta$ and the principal unit axis $\hat{\boldsymbol{e}}$, which are calculated from the initial error quaternion $\boldsymbol{q}^e(t_0) = [q_0^e(t_0), \boldsymbol{q}_v^e(t_0)^T]^T$ as follows:
\begin{equation}
\Delta \theta = 2 \cos^{-1} \big( q_0^e(t_0) \big), \quad \hat{\boldsymbol{e}} = \frac{\boldsymbol{q}_v^e(t_0)}{|\boldsymbol{q}_v^e(t_0)|}
\end{equation}
Using these parameters, the reference trajectory for the spacecraft maneuver is formulated as:
\begin{equation}
\begin{split}\rho(\tau) &=\begin{cases}10\tau^3 - 15\tau^4 + 6\tau^5, & \tau \in [0,1] \\1, & \tau > 1\end{cases}\\
\theta(t) &= \rho\left(\frac{t}{t_s}\right) \Delta \theta \\
\boldsymbol{\omega}^d(t) &= \dot{\theta}(t)\hat{\boldsymbol{e}}, \quad\boldsymbol{q}^d(t) =\begin{bmatrix}\cos(\theta(t)/2) \\\hat{\boldsymbol{e}} \sin(\theta(t)/2)\end{bmatrix}\end{split} \label{eq:Guidance}\end{equation}
The resulting slewing guidance profiles for the attitude and angular velocity are illustrated in Figure \ref{fig:NominalTrajectory}.

\begin{figure}
    \centering
    \includegraphics[width=1\linewidth]{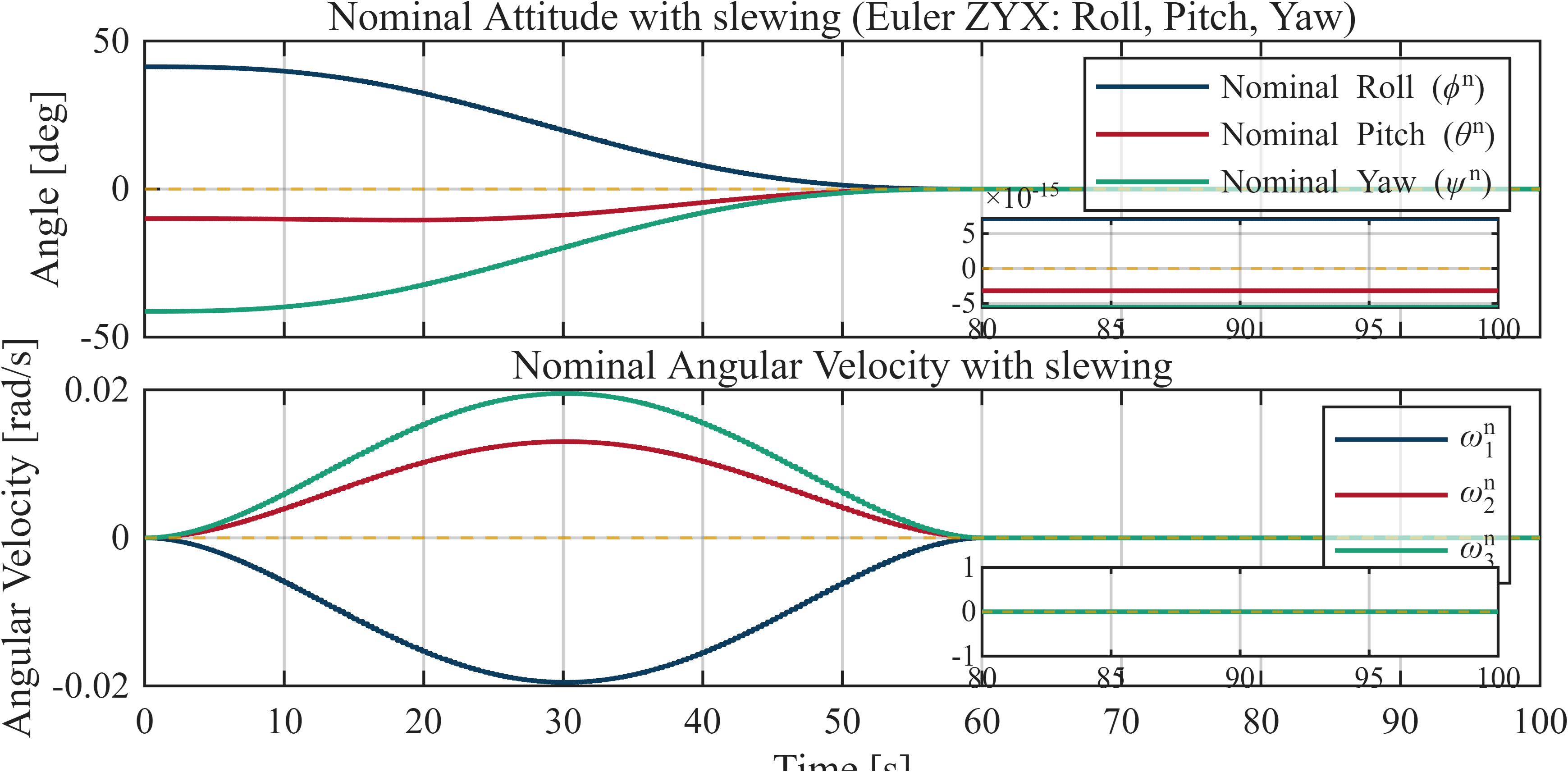}
    \caption{Reference trajectory for the spacecraft slew generated by \eqref{eq:Guidance}: (top) commanded attitude expressed in Euler ZYX angles and (bottom) commanded angular velocity.}
    \label{fig:NominalTrajectory}
\end{figure}
The simulation and controller parameters are set as Table \ref{tab:system_parameters} and Table \ref{tab:controller_parameters}, respectively. There are only five parameters for controller design and timestep adjustment as we can see in Table \ref{tab:controller_parameters}:

\begin{table}[t]
\caption{System Parameters}
\label{tab:system_parameters}
\centering
\renewcommand{\arraystretch}{1.15}
\begin{tabular}{@{}ll@{}}
\toprule
Parameter & Value \\
\midrule
Initial attitude &
$\mathbf{q}^{b}_0=[0.8832,\;0.3,\;-0.2,\;-0.3]^{T}$ \\
Initial angular velocity &
$\boldsymbol{\omega}^{b}_0=[0,\;0,\;0]^{T}\; \mathrm{rad/s}$ \\
Desired attitude &
$\mathbf{q}^{d}=[1,\;0,\;0,\;0]^{T}$ \\
Desired angular velocity &
$\boldsymbol{\omega}^{d}=[0,\;0,\;0]^{T}\; \mathrm{rad/s}$ \\
Slewing/Simulation time &
$t_s = 60\;\mathrm{s}, \quad t \in [0,\;100]\;\mathrm{s}$ \\
Nominal moment of inertia &
$\mathbf{J}_0=\begin{bmatrix}
20 & 1.2 & 0.9 \\
1.2 & 17 & 1.4 \\
0.9 & 1.4 & 15
\end{bmatrix}\mathrm{kg \cdot m^2}$ \\

Inertia uncertainty &
$ \Delta \mathbf{J}=\mathrm{diag}
\begin{bmatrix}
2 & 2 & 1.5
\end{bmatrix}\mathrm{kg \cdot m^2}$ \\

Gravitational parameter &
$\mu=3.986 \times 10^{14}\;\mathrm{m^{3}/s^{2}}$ \\

Orbital radius&
$\|\mathbf{r}_{I}\|=6878.137\;\mathrm{km}$ \\
\bottomrule
\end{tabular}
\end{table}

\begin{table}[t]
\caption{Controller Parameters}
\label{tab:controller_parameters}
\centering
\renewcommand{\arraystretch}{1.15}
\begin{tabular}{@{}ll@{}}
\toprule
Parameter & Value \\
\midrule
Control gain &
$\boldsymbol{\beta} = \mathrm{diag}(3.8,\;3.8,\;3.8)$ \\
Decay factor &
$\gamma = 0.5$ \\
Tolerance &
$\mathrm{tol} = 5 \times 10^{-5}$ \\
Safety factor &
$\mathrm{SF} = 0.98$ \\
Timestep update gain &
$k = 24$ \\
\bottomrule
\end{tabular}
\end{table}

\subsection{Closed-Loop Performance and Theoretical Validation}
\subsubsection{Closed-Loop Performance} Fig. \ref{fig:BodyResult} shows that the spacecraft accurately follows the commanded slew despite inertia uncertainty, environmental disturbances, and time-varying actuator degradation. After the convergence, the attitude and angular-velocity errors are confined to approximately $\pm 2.5 \times 10^{-3}\text{ deg}$ and $\pm 1.5 \times 10^{-5}\text{ rad/s}$. As shown in Fig. \ref{fig:ControlTorque}, the equivalent and robust control components generate a bounded total control torque of approximately $\pm0.04 \textrm{ N} \cdot \textrm{m}$. The proposed scheduling law simultaneously adapts the control interval to the local TDE level, yielding a mean timestep of $0.35297$ s while maintaining the above tracking performance. The resulting timestep behavior is examined in Fig. \ref{fig:TDEPerformance}

\begin{figure}
    \centering
    \includegraphics[width=1\linewidth]{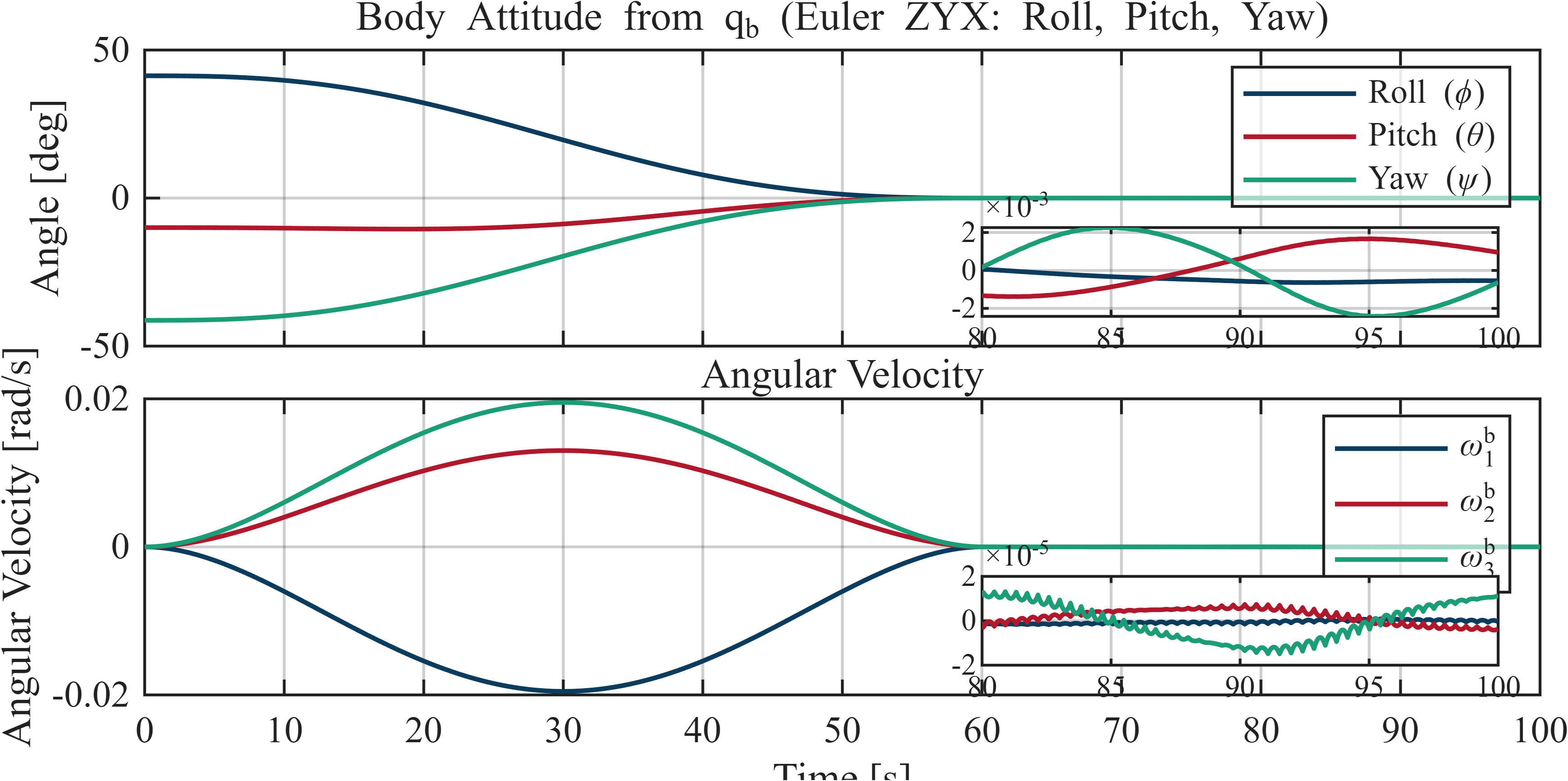}
    \caption{Closed-loop spacecraft response under VTDC: (top) body attitude expressed in Euler ZYX angles and (bottom) body angular velocity. Insets show the post-slew steady-state response.}
    \label{fig:BodyResult}
\end{figure}

\begin{figure}
    \centering
    \includegraphics[width=1\linewidth]{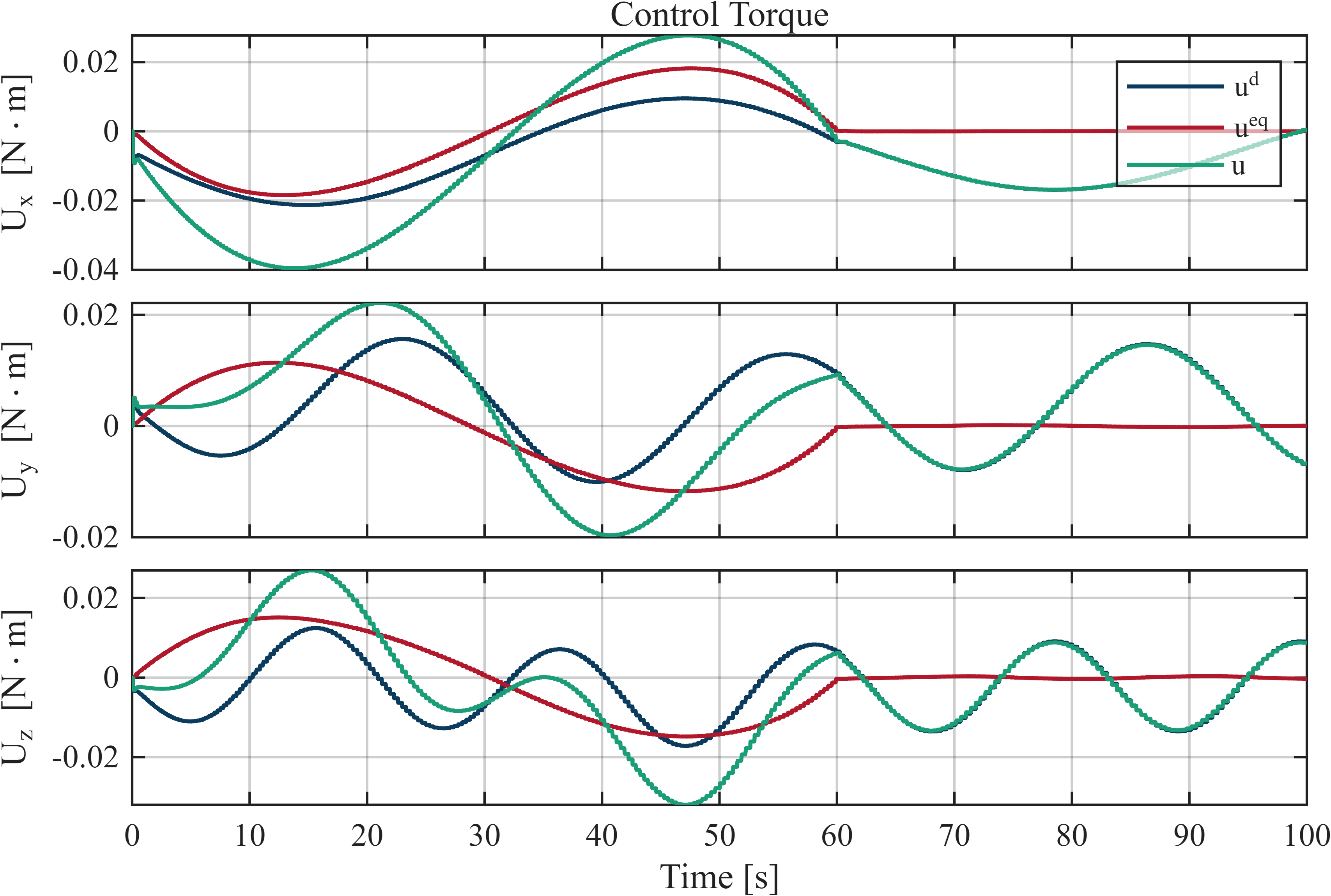}
    \caption{Control torque decomposition under VTDC for each body axis, showing the equivalent control $\boldsymbol{u}^{\text{eq}}$, robust time-delay compensation $\boldsymbol{u}^{\text{d}}$, and total control input $\boldsymbol{u}$.}
    \label{fig:ControlTorque}
\end{figure}

\subsubsection{TDE Characterization and Estimation} 
Fig.~\ref{fig:TDEValidation} assesses the practical System TDE estimator in \eqref{eq:SystemTDEEstimation} using the actual model and disturbance information available in simulation. The estimated model, disturbance, and total TDEs closely follow their actual counterparts, while the estimation margins remain strictly positive and on the order of $10^{-6}$ to $10^{-5}$ throughout the maneuver. This supports the use of the proposed estimator as a conservative scheduling surrogate under the considered operating conditions.
 
\begin{figure}
    \centering
    \includegraphics[width=1\linewidth]{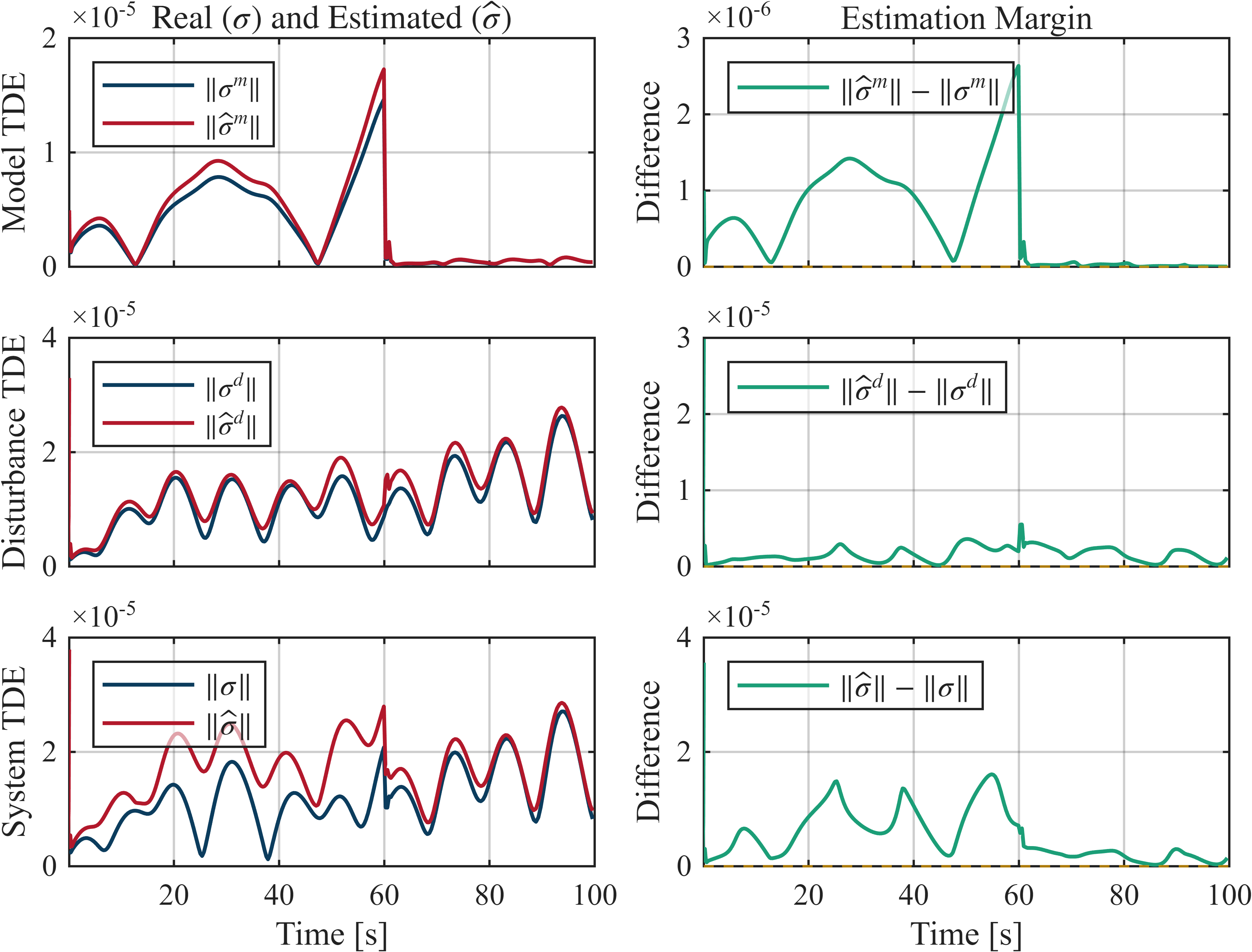}
    \caption{Assessment of the practical System-TDE estimator in \eqref{eq:SystemTDEEstimation}: (left) actual and estimated magnitudes of the model, disturbance, and total System TDEs; (right) corresponding estimation margins, showing conservative estimation under the considered operating conditions.}
    \label{fig:TDEValidation}
\end{figure}

The timestep dependence predicted by Lemma~\ref{lemma:BoundnessOfError} was further examined using fixed-step simulations with $h=10^{-4}, 10^{-3}, 10^{-2}, 10^{-1}$ s. Fig.~\ref{fig:TDE_IAE_perTimestep} shows that the maximum System TDE increases by approximately two orders of magnitude for each tenfold increase in $h$, consistent with the $\mathcal{O}(h^2)$ bound in \eqref{eq:BoundnessOfError}. The quaternion-error Integral Absolute Error (IAE) exhibits the same empirical trend, illustrating the associated degradation in tracking accuracy.

\begin{figure}
    \centering
    \includegraphics[width=1\linewidth]{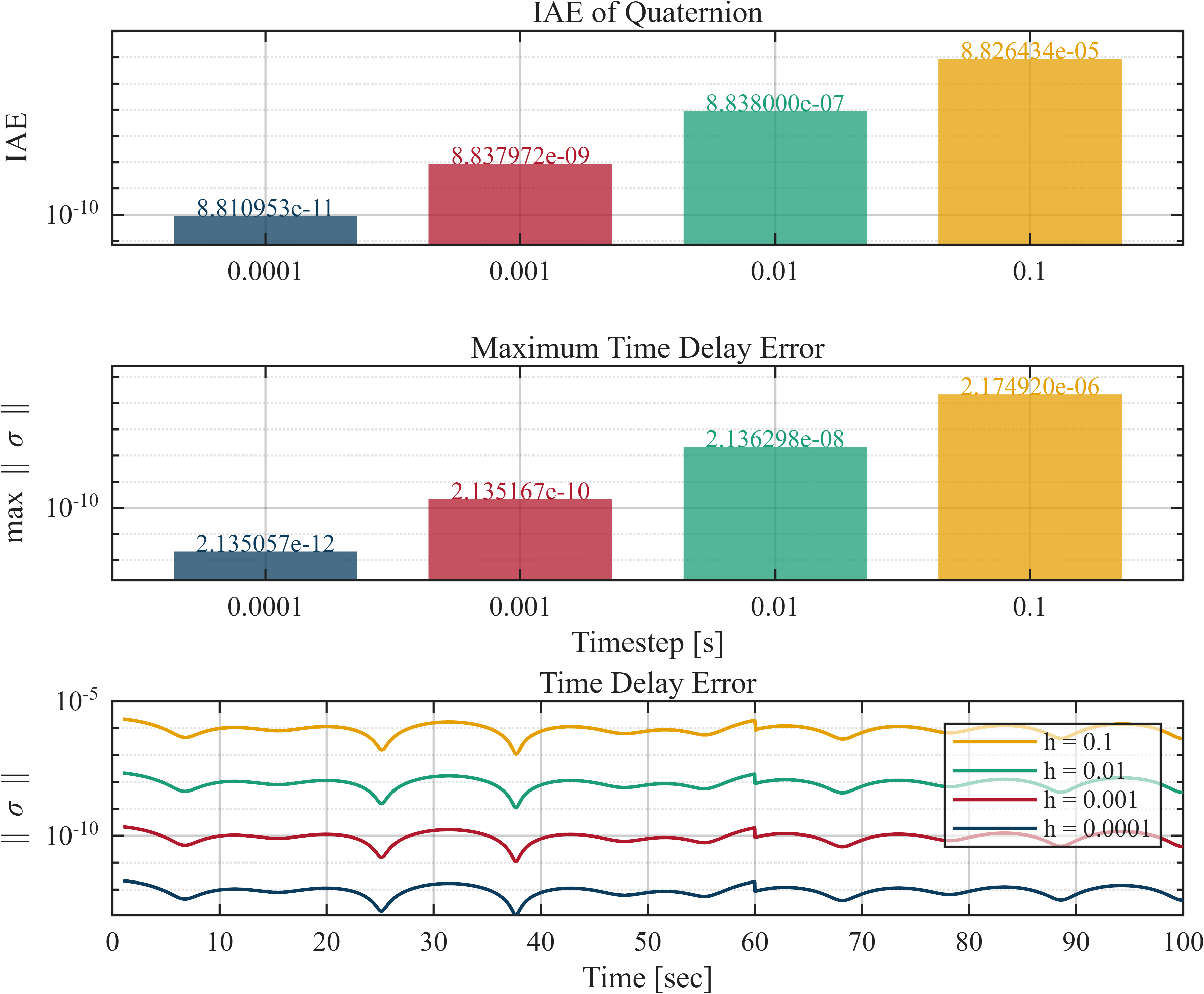}
    \caption{Fixed-step assessment of the timestep dependence for $h=10^{-4}$, $10^{-3}$, $10^{-2}$, and $10^{-1}$~s: (top) quaternion-error IAE, (middle) maximum System TDE, and (bottom) System-TDE histories. The maximum TDE exhibits the $\mathcal{O}(h^{2})$ scaling predicted by Lemma \ref{lemma:BoundnessOfError}, while the quaternion-error IAE shows the same empirical trend.}
    \label{fig:TDE_IAE_perTimestep}
\end{figure}

\subsubsection{Validation of Timestep and Stability Properties} 
The initial-timestep dependence in Corollary \ref{cor:initial_independence} was assessed through 100 Monte Carlo simulations with $h_0\in[0.1, 1]$ s. As shown in Fig. \ref{fig:InitialConditionIndep}, the initially broad timestep distribution rapidly contracts, and its standard deviation decreases from approximately $0.25$ s to $10^{-3}$ s, confirming strong attenuation of the initial-timestep influence.

\begin{figure}
    \centering
    \includegraphics[width=1\linewidth]{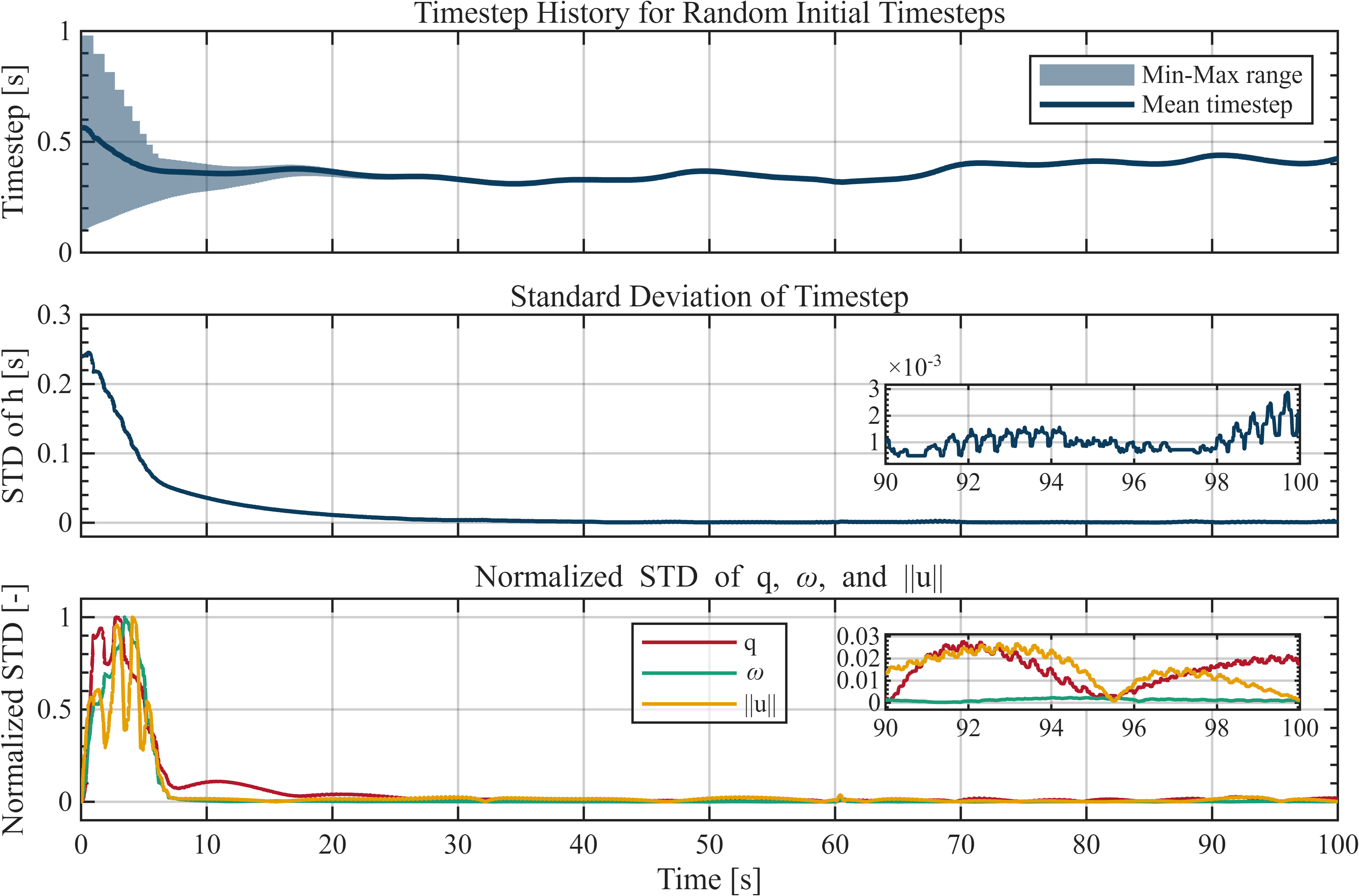}
    \caption{Monte Carlo assessment of initial-timestep independence for $N_{\mathrm{MC}}=100$ and $h_0\in[0.1, 1]$~s:(a) mean timestep with min–max envelope, (b) timestep standard deviation, and (c) normalized standard deviations of $q$, $\omega$, and $\|u\|$, showing attenuation of the initial timestep-induced dispersion.}
    \label{fig:InitialConditionIndep}
\end{figure}

For the post-transient interval $t\geq15$ s, the observed TDE coefficients from $\boldsymbol{f}_n$ and $\boldsymbol{d}_n$ give $\hat{\Psi} =2.7858\times10^{-4}$ and $\hat{\psi}_{\min}=4.9566\times10^{-5}$. Corollaries \ref{cor:NoZeno} and \ref{cor:StepRatioBound} therefore predict

\begin{equation}
\begin{split}
    h_{\min} = 0.26088~\mathrm{s},\quad h_{\max}=0.61847~\mathrm{s}\\
    r_{\min} = 0.96467, \qquad r_{\max} = 1.03662.
\end{split}
\end{equation}

Fig. \ref{fig:TDEPerformance} shows that the realized post-transient timestep, $h_n\in[0.317,0.450]$ s, and the corresponding step ratios remain within these predicted ranges while the System TDE remains below the prescribed tolerance.

\begin{figure}
    \centering
    \includegraphics[width=1\linewidth]{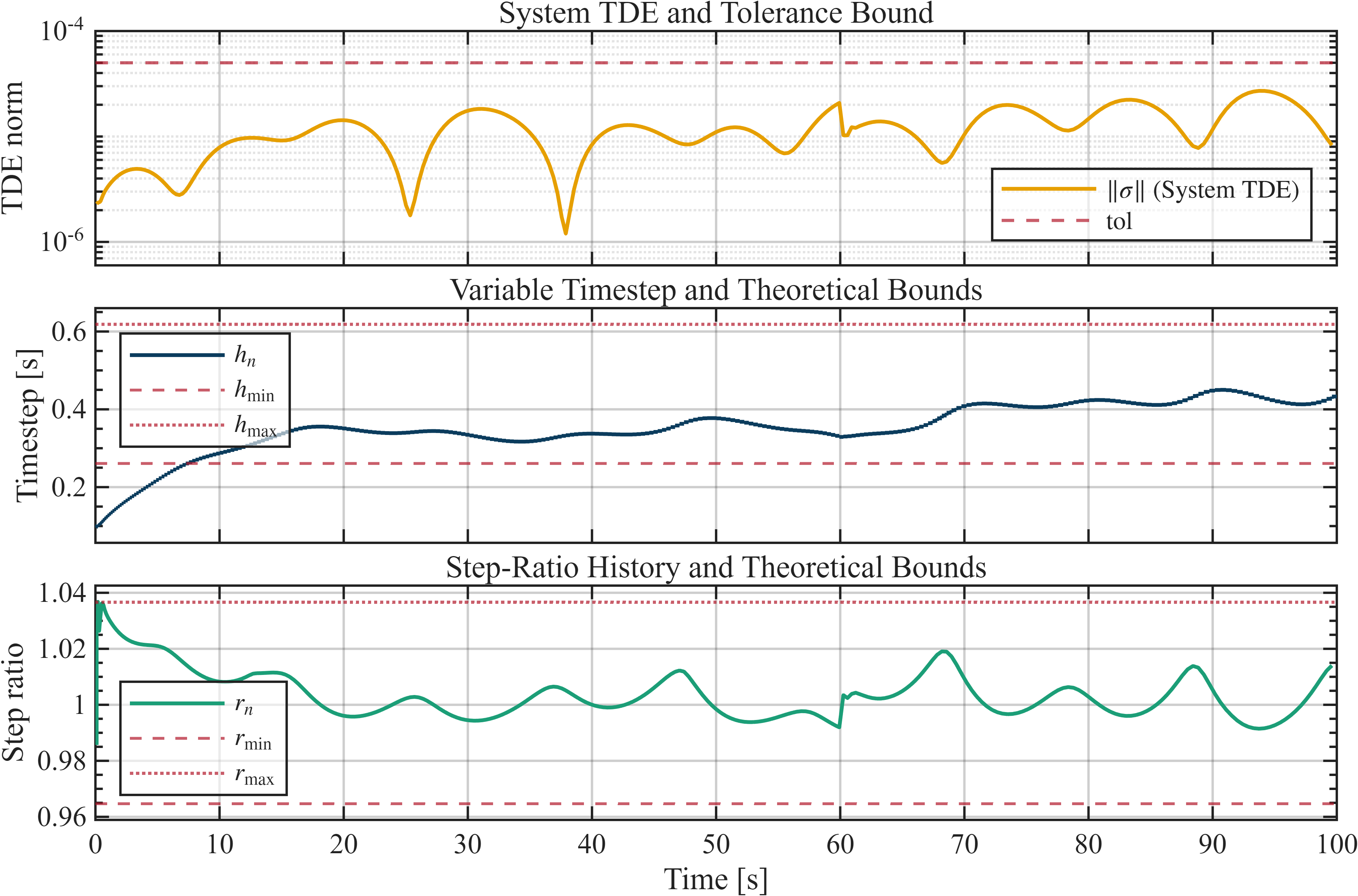}
    \caption{Consistency of the realized TDE, timestep, and step ratio with the derived scheduling properties: (top) System TDE and prescribed tolerance, (middle) variable timestep $h_n$, and its post-transient bounds $h_{\min}$ and $h_{\max}$, and (bottom) step ratio $r_n$ and its bounds $r_{\min}$ and $r_{\max}$. The theoretical bounds are evaluated over post-transient interval $t\geq15$~s}.
    \label{fig:TDEPerformance}
\end{figure}

Using the theoretical step-ratio bounds gives $\bar{\gamma}=0.51766<1$ satisfying the stability condition of Theorem 2. Neglecting the higher-order $\mathcal{O}(h^2)$ residual terms, the leading-order general and steady-state convergence radii are:
\begin{equation}
    \frac{r_{\max}}{1-\bar{\gamma}}\cdot\mathrm{tol} =  1.0746\times10^{-4},\qquad
    \frac{\mathrm{tol}}{1-|\gamma|} =    1.0000\times10^{-4}.
\end{equation}

Fig. \ref{fig:BoundedSV} shows that the sliding-variable norm remains below even these leading-order reference radii throughout the simulation. Since the complete UUB bounds in Theorem \ref{Theo:UUB} and Corollary \ref{cor:SteadyStateBound} additionally contain bounded $\mathcal{O}(h^2)$ residual terms, the observed response is consistent with the derived stability result and indicates that the higher-order contribution is small under the considered operating conditions.

\begin{figure}
    \centering
    \includegraphics[width=1\linewidth]{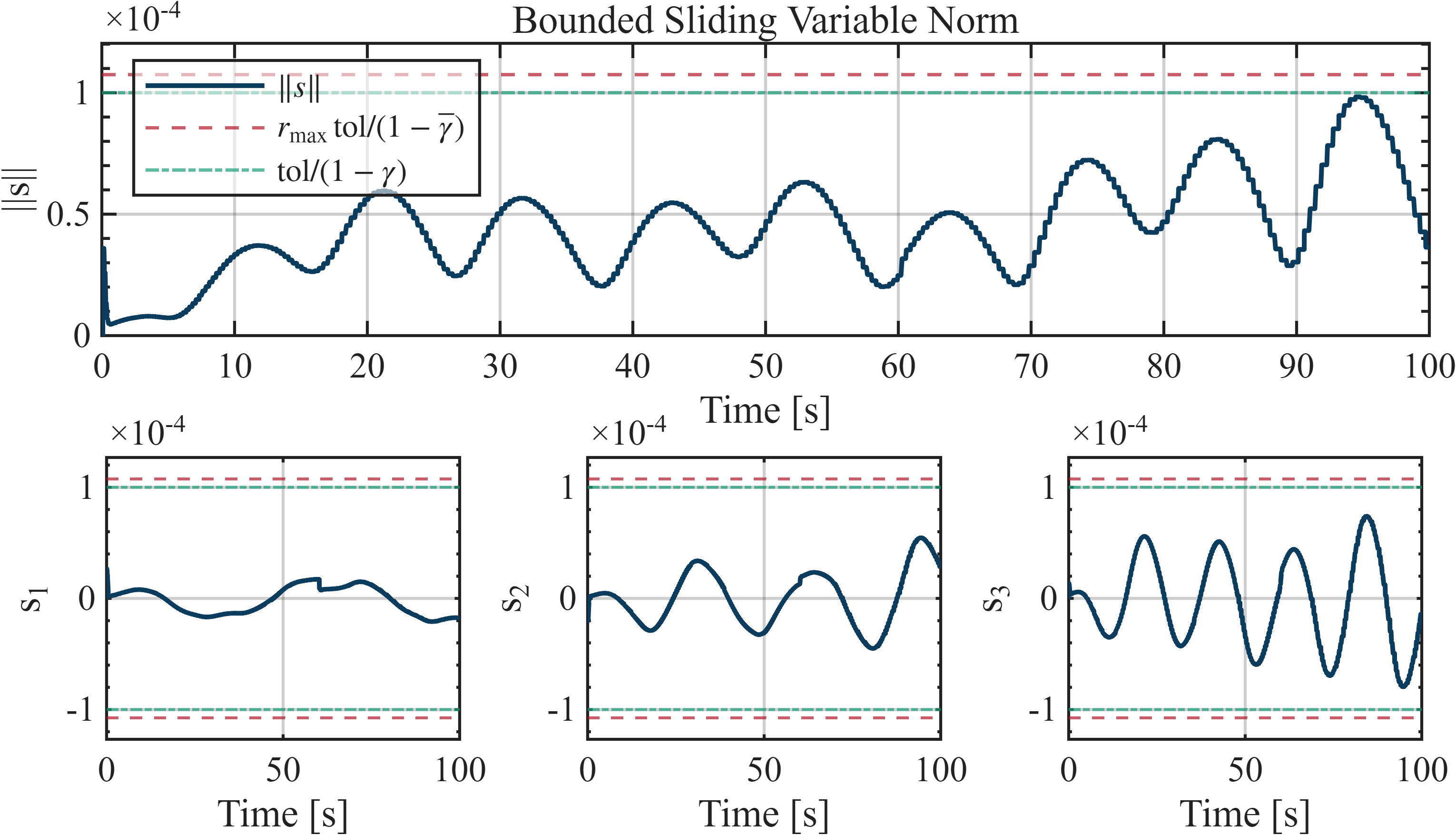}
    \caption{Sliding-variable response relative to the leading-order convergence radii: (top) $\|\boldsymbol{s}\|$ with the general and steady-state leading-order reference radii, and (bottom) individual sliding-variable components $s_1$, $s_2$ and $s_3$. The reference radii omit the bounded $\mathcal{O}(h^2)$ residual terms in Theorem \ref{Theo:UUB} and Corollary~\ref{cor:SteadyStateBound}}.
    \label{fig:BoundedSV}
\end{figure}


\subsection{Comparative Evaluation of Control and Scheduling Performance}
\begin{table*}[t]
\centering
\caption{Comparison of Tracking Error and Control Effort among STC, VTDC, and FTDC}
\label{tab:performance_comparison}
\scriptsize
\setlength{\tabcolsep}{4pt}
\renewcommand{\arraystretch}{1.15}
\begin{tabular}{c|ccc|ccc|ccc}
\hline
\multirow{2}{*}{Method}
& \multicolumn{3}{c|}{Attitude Error}
& \multicolumn{3}{c|}{Angular Velocity Error}
& \multicolumn{3}{c}{Control Effort} \\
\cline{2-10}
& IAE & ISE & ITAE
& IAE & ISE & ITAE
& Max Torque & IAC & ISC \\
\hline
STC
& $1.1836\times10^{-3}$ & $2.2856\times10^{-8}$ &$6.5074\times10^{-2}$
& $1.0556\times10^{-3}$ & $1.9343\times10^{-8}$ & $5.4860\times10^{-2}$
& $5.0146\times10^{-2}$ & $2.3852$ & $7.3162\times10^{-2}$ \\
VTDC
& $1.1102\times10^{-3}$ & $1.2368\times10^{-8}$ & $6.5405\times10^{-2}$
& $6.5684\times10^{-4}$ & $5.0111\times10^{-9}$ & $3.7210\times10^{-2}$
& $5.0176\times10^{-2}$ & $2.3822$ & $7.3115\times10^{-2}$ \\
FTDC
& $1.0218\times10^{-3}$ & $9.3963\times10^{-9}$ & $5.1869\times10^{-2}$
& $6.4220\times10^{-4}$ & $5.9835\times10^{-9}$ & $2.9756\times10^{-2}$
& $5.0169\times10^{-2}$ & $2.3822$ & $7.3149\times10^{-2}$ \\
\hline
\end{tabular}
\end{table*}

\begin{table*}[t]
\centering
\caption{Comparison of Timestep Behavior and Timestep Computation Cost between STC and VTDC}
\label{tab:timestep_computation_comparison}
\scriptsize
\setlength{\tabcolsep}{4pt}
\renewcommand{\arraystretch}{1.15}
\begin{tabular}{c|cccc|ccccc}
\hline
\multirow{2}{*}{Method}
& \multicolumn{4}{c|}{Timestep Behavior}
& \multicolumn{4}{c}{Timestep Computation Cost} \\
\cline{2-9}
& Mean $h$ [s] & Max $h$ [s] & Min $h$ [s] & Iterations
& Mean [s] & Max [s] & Min [s] & Total [s]  \\
\hline
STC
& $0.33165$ & $0.886$ & $0.100$ & $343$
& $7.6163\times10^{-2}$ & $3.1444\times10^{-1}$ & $4.5463\times10^{-2}$ & $26.048$  \\
VTDC
& $0.35297$ & $0.450$ & $0.0985$ & $299$
& $4.2032\times10^{-5}$ & $3.7600\times10^{-4}$ & $3.0700\times10^{-5}$ & $1.2525\times10^{-2}$  \\
\hline
\end{tabular}
\end{table*}

\begin{figure}
    \centering
    \includegraphics[width=1\linewidth]{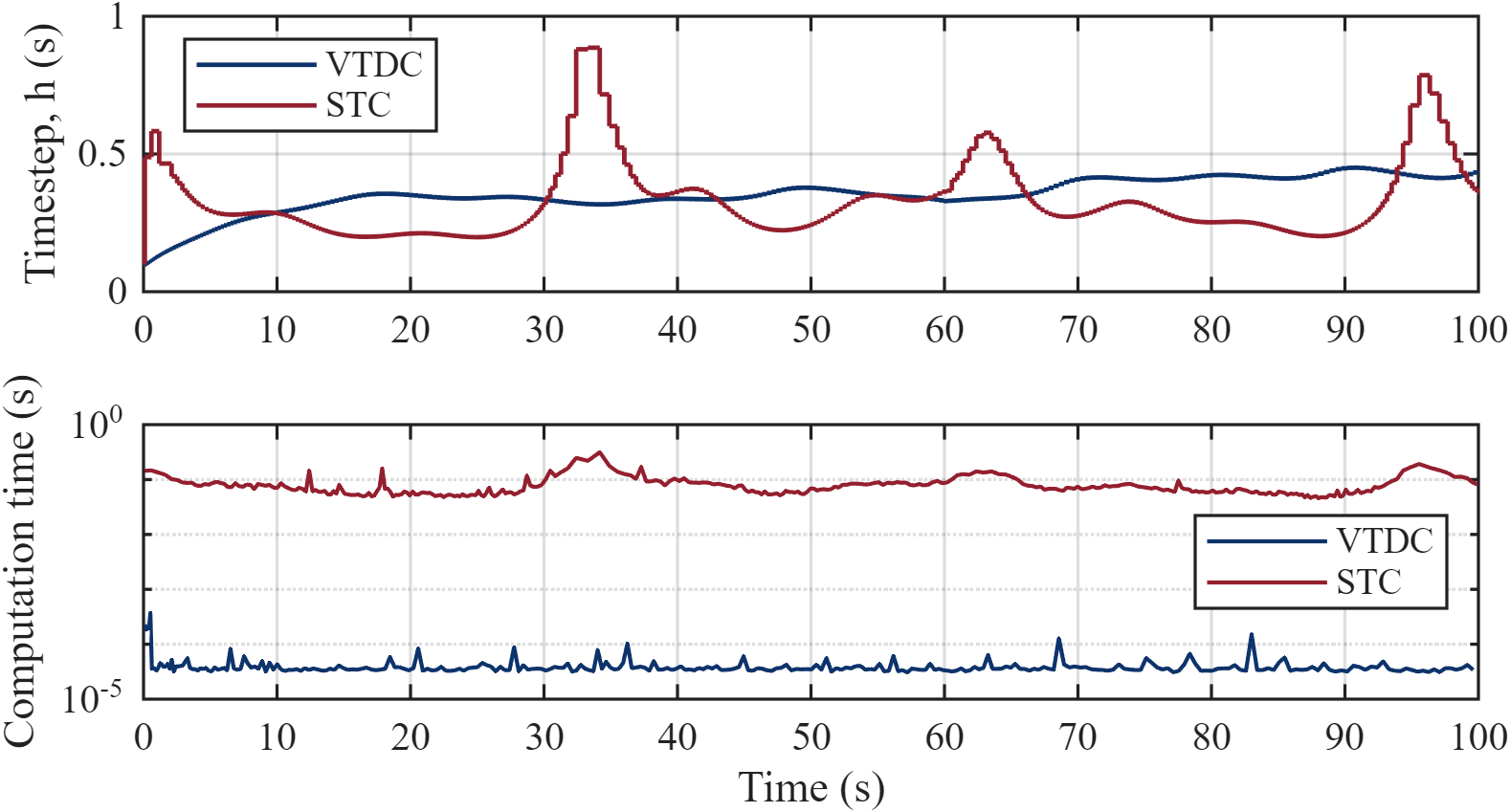}
    \caption{Comparison of STC and VTDC scheduling behavior: (top) control-interval histories and (bottom) computation time required to determine each subsequent control interval.}
    \label{fig:stc_vtdc_timestep_computation}
\end{figure}

The proposed VTDC was compared with a Lyapunov-based self-triggered controller (STC) and fixed-step TDC (FTDC). All methods were evaluated under identical spacecraft dynamics, initial conditions, reference trajectories, uncertainties, disturbances, and actuator degradation. The same TDC-based control law in \eqref{eq:TotalControl} and zero-order hold implementation were used; therefore, the primary difference among the methods is the determination of the control-update interval.

VTDC determines the next timestep directly from the estimated System TDE using \eqref{eq:UpdateRule}. For comparison, a Lyapunov-based STC benchmark was formulated following the self-triggered implementation principle in \cite{heemels2012introduction}. At each update instant $t_k$, the nonlinear spacecraft dynamics are numerically propagated under the held control input, and the future sliding variable $\boldsymbol{s}(t_k+\tau)$ is evaluated over candidate intervals. The largest admissible interval is selected such that
\begin{equation}
V(\boldsymbol{s}(t_k+\tau))\leq V(\boldsymbol{s}(t_k)) e^{-\delta \tau}+\|\hat w\|.
\end{equation}

The STC parameters were set to $\delta=5$ and $\|\hat w\|=6\times10^{-8}$, with the candidate interval searched numerically in increments of $10^{-4}$~s. FTDC used the VTDC mean timestep,  $h=0.35297$~s, providing a fixed-step reference with approximately the same average control-update rate.

The performance comparison was conducted in terms of tracking accuracy, control effort, update efficiency, and scheduling cost. Tracking performance was evaluated using the Integral Absolute Error (IAE), Integral Squared Error (ISE), and Integral Time-weighted Absolute Error (ITAE), while control effort was assessed using the maximum torque, Integral Absolute Control (IAC), and Integral Squared Control (ISC). IAE represents the accumulated tracking error, ISE emphasizes large deviations, and ITAE penalizes persistent errors, whereas IAC and ISC quantify the overall and squared control usage, respectively.
\begin{equation}
\begin{split}
&\mathrm{IAE} =\int |e| dt,\quad\mathrm{ISE} =\int e^2 dt,\quad \mathrm{ITAE} =\int t\cdot|e| dt\\
&u_{\max} =\max_i{\|\boldsymbol{u}^c_i\|},\;
\mathrm{IAC} =\sum_i{\|\boldsymbol{u}^c_i\|}h_i,\; \mathrm{ISC} =\sum_i{\|\boldsymbol{u}^c_i\|^2}h_i
\end{split}
\end{equation}
Table \ref{tab:performance_comparison} compares tracking error and control effort using IAE, ISE, and ITAE for the attitude and angular-velocity errors, together with maximum torque, IAC, and ISC. All three methods produced nearly identical control-effort metrics, as expected from their common TDC-based feedback structure. Compared with STC, VTDC reduced the attitude IAE and ISE by approximately $6.2\%$ and $45.9\%$, respectively, and reduced the angular-velocity IAE, ISE, and ITAE by $37.8\%$, $74.1\%$, and $32.2\%$, respectively. Thus, the reduced update frequency of VTDC did not compromise closed-loop performance.

FTDC achieved smaller values for several tracking metrics under the considered simulation condition. This result is reasonable because FTDC uses the mean timestep obtained from the VTDC simulation, $h=0.35297$~s, and therefore operates at a fixed interval identified a posteriori from the VTDC scheduling behavior without additional performance variation caused by timestep adaptation. The comparison thus does not imply that a fixed timestep is intrinsically superior; rather, it suggests that VTDC can also serve as a practical means of identifying an effective fixed control interval for a given operating condition. Unlike FTDC, however, VTDC retains the ability to adapt this interval online in response to changes in the local TDE and system variation.

For onboard spacecraft control, the principal advantage of VTDC lies in its scheduling efficiency. As summarized in Table~\ref{tab:timestep_computation_comparison}, VTDC increased the mean control interval from $0.33165$~s for STC to $0.35297$~s and reduced the number of control updates from $343$ to $299$, corresponding to a $6.4\%$ longer mean interval and $12.8\%$ fewer updates. More importantly, the accumulated timestep-computation time decreased from $26.048$~s for STC to $0.012525$~s for VTDC, representing an approximately $2.08\times10^3$-fold reduction in onboard scheduling overhead.

Fig. \ref{fig:stc_vtdc_timestep_computation} illustrates this difference directly. STC requires repeated integration to determine each triggering instant, whereas VTDC evaluates the next timestep through the closed-form TDE-based update rule. Consequently, VTDC retains proactive aperiodic scheduling while avoiding the computational burden associated with iterative future-state prediction.

\section{Conclusion}
This paper developed Variable-step Time Delay Control (VTDC), a proactive aperiodic robust control framework integrating TDC-based uncertainty compensation and control-update scheduling. By establishing the quadratic dependence of the System TDE on the realized timestep, a closed-form feedback law was derived to adapt the control interval according to local model and disturbance variations.

The variable-step closed loop was shown to admit asymptotically bounded timesteps and step ratios while attenuating the influence of the initial timestep and excluding Zeno behavior. These properties were incorporated into the stability analysis to establish uniform ultimate boundedness of the sliding variable and characterize its steady-state convergence radius.

A practical System-TDE surrogate constructed from equivalent- and robust-control variations enables algebraic timestep scheduling without direct disturbance measurement, continuous trigger monitoring, future-state prediction, or iterative trigger-time search.

Spacecraft attitude simulations under inertia uncertainty, disturbance torques, and actuator degradation demonstrated accurate tracking and adaptive scheduling. Compared with Lyapunov-based STC, VTDC reduced control updates from $343$ to $299$ and timestep-scheduling computation from $26.048$~s to $0.012525$~s. The fixed-step comparison further indicates that the timestep profile can also serve as a reference for selecting an effective constant control interval when adaptive scheduling is unnecessary.

\appendices
\section{Proof of Corollary \ref{cor:initial_independence}: Initial Timestep Independence}\label{app:initial_independence}
Substituting the definition of $\|\boldsymbol{\sigma}_n\| = \psi_n h_n^2$ in \eqref{eq:DiscreteBound1} into the update rule in \eqref{eq:UpdateRule} yields the following recurrence relation for time step $h$. 
\begin{equation}
h_{n+1} = \text{SF} \cdot \Bigg( \frac{\text{tol}}{\psi _nh_n^2} \Bigg)^{\frac{1}{2k}} h_n = \text{SF} \cdot \Big( \frac{\text{tol}}{\psi_n} \Big)^{\frac{1}{2k}} h_n^{1-\frac{1}{k}} \label{eq:hrecurr}
\end{equation}
By substituting $\lambda = 1 - \frac{1}{k}\in(0, 1)$ and taking the natural logarithm of both sides, we obtain the following non-homogeneous linear difference equation for $\ln h_n$:
\begin{equation}
\ln h_{n+1} = \lambda\ln h_n + \underbrace{\ln{\Bigg(\text{SF} \cdot \Big( \frac{\text{tol}}{\psi_n} \Big)^{\frac{1}{2k}}\Bigg)}}_{C_n} \label{eq:hRecurrence}
\end{equation}

Here, $C_n$ represents the time-varying input driven by the actual system state $\psi_n$. The general solution for the first-order difference equation $x_{n+1} = \lambda x_n + C_n$ is $x_{n} = \lambda^n x_0 + \Sigma ^{n-1}_{i=0}\lambda^{n-1-i}C_i$. Using this, the general solution for $\ln{h_n}$ in \eqref{eq:hRecurrence} is given as follows:

\begin{equation}
\ln h_n = \lambda^n \ln h_0 + \sum_{i=0}^{n-1} \lambda^{n-1-i} C_i \label{eq:h_lim}
\end{equation}

Since $0<\lambda< 1$. As the number of iterations $n \to \infty$, the term involving the initial timestep, $\lambda^n \ln h_0$, decays exponentially to zero.
\begin{equation}\label{eq:hConverge}
\lim_{n \to \infty}\ln h_n = \underbrace{\lim_{n \to \infty}\lambda^n \ln h_0}_{=0} + \lim_{n \to \infty}\sum_{i=0}^{n-1} \lambda^{n-1-i} C_i\\
\end{equation}
This mathematical property demonstrates that the transient influence of the initial heuristic setting $h_0$ vanishes over time. This shows that the asymptotic timestep sequence is governed by the accumulated coefficient history ($C_i$), while its dependence on the initial heuristic timestep $h_0$ vanishes.

\section{Proof of Corollary \ref{cor:NoZeno}: Timestep Boundedness and Exclusion of Zeno Behavior}\label{app:timestep_properties}
By Assumption \ref{assumption:Lipschitz} and Assumption \ref{assumption:PersistentTDE}, the condition $0<\psi_{\min}\leq\psi_i\leq\Psi$ holds for all steps $i$. Therefore, $C_i$ from \eqref{eq:hRecurrence} is uniformly bounded as 
\begin{equation} 
C_{\min} \leq C_i=\ln{\Bigg(\text{SF} \cdot \Big( \frac{\text{tol}}{\psi_i} \Big)^{\frac{1}{2k}}\Bigg)} \leq C_{\max}, 
\end{equation} 
with 
\begin{equation} C_{\min} = \ln\Bigg[\mathrm{SF} \Big( \frac{\mathrm{tol}}{\Psi}\Big)^{\frac{1}{2k}} \Bigg], \quad C_{\max} = \ln\Bigg[\mathrm{SF} \Big( \frac{\mathrm{tol}}{\psi_{\min}}\Big)^{\frac{1}{2k}} \Bigg] \end{equation}

Note that the general solution for $\ln{h_n}$ is given as \eqref{eq:h_lim}. Using the bounds on $C_i$ gives \begin{equation} 
\begin{split}
    \ln h_n &\geq \lambda^n \ln h_0 + C_{\min} \sum_{i=0}^{n-1} \lambda^{n-1-i}\\
    \ln h_n &\leq \lambda^n \ln h_0 + C_{\max} \sum_{i=0}^{n-1} \lambda^{n-1-i} 
\end{split}
\end{equation}

The summation term $\sum_{i=0}^{n-1} \lambda^{n-1-i}$ represents a finite geometric series. By reversing the index of summation (letting $j = n-1-i$), it can be rewritten and evaluated as follows:
\begin{equation}
\sum_{i=0}^{n-1} \lambda^{n-1-i} = \sum_{j=0}^{n-1} \lambda^j = \frac{1 - \lambda^n}{1 - \lambda}
\end{equation}

The term $\lambda^n$ exponentially decays to zero as $n \to \infty$ because of the condition $ \lambda\in(0,1)$. Consequently, the infinite geometric series converges to $\frac{1}{1-\lambda}=k$, which yields:
\begin{equation}
kC_{\min}\leq\liminf_{n\to\infty}\ln h_n\leq\limsup_{n\to\infty}\ln h_n \leq kC_{\max}
\end{equation}

Finally, taking the exponential of both sides yields 
\begin{equation} \liminf_{n\to\infty} h_n \geq \mathrm{SF}^k \sqrt{\frac{\mathrm{tol}}{\Psi}}, \quad \limsup_{n\to\infty} h_n \leq \mathrm{SF}^k \sqrt{\frac{\mathrm{tol}}{\psi_{\min}}}, 
\end{equation} 
which proves \eqref{eq:TimestepUltimateBound}.



\section*{Acknowledgment}

This work was supported by the National Research Foundation of Korea(NRF) grant funded by the Korea government(MSIT) (RS-2025-00560240).

\ifCLASSOPTIONcaptionsoff
  \newpage
\fi



\bibliographystyle{IEEEtran}
\bibliography{References/All}

@book{franklin1998digital,
  title={Digital control of dynamic systems},
  author={Franklin, Gene F and Powell, J David and Workman, Michael L and others},
  volume={3},
  year={1998},
  publisher={Addison-wesley Menlo Park, CA}
}

@article{wittenmark2002computer,
  title={Computer control: An overview},
  author={Wittenmark, Bj{\"o}rn and {\AA}rz{\'e}n, Karl-Erik and {\AA}str{\"o}m, Karl Johan},
  journal={IFAC Professional Brief},
  year={2002},
  publisher={International Federation of Automatic Control}
}

@article{fridman2014introduction,
  title={Introduction to time-delay systems},
  author={Fridman, Emilia},
  journal={Analysis and Control. Birkh{\"a}user},
  volume={75},
  pages={102},
  year={2014},
  publisher={Springer}
}

@article{fridman2005input,
  title={Input/output delay approach to robust sampled-data H$\infty$ control},
  author={Fridman, Emilia and Shaked, Uri and Suplin, Vladimir},
  journal={Systems \& Control Letters},
  volume={54},
  number={3},
  pages={271--282},
  year={2005},
  publisher={Elsevier}
}

@article{seuret2012novel,
  title={A novel stability analysis of linear systems under asynchronous samplings},
  author={Seuret, Alexandre},
  journal={Automatica},
  volume={48},
  number={1},
  pages={177--182},
  year={2012},
  publisher={Elsevier}
}

@inproceedings{heemels2012introduction,
  title={An introduction to event-triggered and self-triggered control},
  author={Heemels, Wilhelmus PMH and Johansson, Karl Henrik and Tabuada, Paulo},
  booktitle={2012 ieee 51st ieee conference on decision and control (cdc)},
  pages={3270--3285},
  year={2012},
  organization={IEEE}
}

@book{fortescue2011spacecraft,
  title={Spacecraft systems engineering},
  author={Fortescue, Peter and Swinerd, Graham and Stark, John},
  year={2011},
  publisher={John Wiley \& Sons}
}

@article{di2021event,
  title={Event-triggered sliding mode attitude coordinated control for spacecraft formation flying system with disturbances},
  author={Di, Fu-Qiang and Li, Ai-Jun and Guo, Yong and Xie, Cheng-Qing and Wang, Chang-Qing},
  journal={Acta Astronautica},
  volume={188},
  pages={121--129},
  year={2021},
  publisher={Elsevier}
}

@article{xie2024dynamic,
  title={Dynamic event-triggered and self-triggered fault-tolerant attitude control for multiple spacecraft systems with uncertainties and input saturation},
  author={Xie, Xiong and Sheng, Tao and Chen, Xiaoqian},
  journal={IEEE Transactions on Aerospace and Electronic Systems},
  volume={60},
  number={3},
  pages={2922--2933},
  year={2024},
  publisher={IEEE}
}

@article{de2020self,
  title={Self-triggered output-feedback control of LTI systems subject to disturbances and noise},
  author={de Albuquerque Gleizer, Gabriel and Mazo Jr, Manuel},
  journal={Automatica},
  volume={120},
  pages={109129},
  year={2020},
  publisher={Elsevier}
}

@article{anta2010sample,
  title={To sample or not to sample: Self-triggered control for nonlinear systems},
  author={Anta, Adolfo and Tabuada, Paulo},
  journal={IEEE Transactions on automatic control},
  volume={55},
  number={9},
  pages={2030--2042},
  year={2010},
  publisher={IEEE}
}

@article{delimpaltadakis2021region,
  title={Region-based self-triggered control for perturbed and uncertain nonlinear systems},
  author={Delimpaltadakis, Giannis and Mazo, Manuel},
  journal={IEEE Transactions on Control of Network Systems},
  volume={8},
  number={2},
  pages={757--768},
  year={2021},
  publisher={IEEE}
}

@book{utkin2013sliding,
  title={Sliding modes in control and optimization},
  author={Utkin, Vadim I},
  year={2013},
  publisher={Springer Science \& Business Media}
}

@article{gao1995discrete,
  title={Discrete-time variable structure control systems},
  author={Gao, Weibing and Wang, Yufu and Homaifa, Abdollah},
  journal={IEEE transactions on Industrial Electronics},
  volume={42},
  number={2},
  pages={117--122},
  year={1995},
  publisher={IEEE}
}

@article{cho2020autonomous,
  title={Autonomous precision control of satellite formation flight under unknown time-varying model and environmental uncertainties},
  author={Cho, Hancheol and Udwadia, Firdaus E and Wanichanon, Thanapat},
  journal={The Journal of the Astronautical Sciences},
  volume={67},
  number={4},
  pages={1470--1499},
  year={2020},
  publisher={Springer}
}

@article{youcef1990time,
    author = {Youcef-Toumi, Kamal and Ito, Osamu},
    title = {A Time Delay Controller for Systems With Unknown Dynamics},
    journal = {Journal of Dynamic Systems, Measurement, and Control},
    volume = {112},
    number = {1},
    pages = {133-142},
    year = {1990},
    month = {03},
    issn = {0022-0434},
    doi = {10.1115/1.2894130},
    eprint = {https://asmedigitalcollection.asme.org/dynamicsystems/article-pdf/112/1/133/5557301/133_1.pdf},
}

@article{drazenovic1969,
  author  = {Dra{\v z}enovi{\'c}, B.},
  title   = {The Invariance Conditions in Variable Structure Systems},
  journal = {Automatica},
  volume  = {5},
  number  = {3},
  pages   = {287--295},
  year    = {1969},
  doi     = {10.1016/0005-1098(69)90071-5}
}

@article{furuta1990,
title = {Sliding mode control of a discrete system},
journal = {Systems and Control Letters},
volume = {14},
number = {2},
pages = {145-152},
year = {1990},
issn = {0167-6911},
doi = {https://doi.org/10.1016/0167-6911(90)90030-X},
url = {https://www.sciencedirect.com/science/article/pii/016769119090030X},
author = {Katsuhisa Furuta}
}

@article{milosavljevic1985,
  author  = {Milosavljevi{\'c}, {\v C}edomir},
  title   = {General Conditions for the Existence of a Quasi-Sliding Mode on the Switching Hyperplane in Discrete Variable Structure Systems},
  journal = {Automation and Remote Control},
  volume  = {46},
  number  = {3},
  pages   = {307--314},
  year    = {1985}
}

@ARTICLE{utkin1977,
  author={Utkin, V.},
  journal={IEEE Transactions on Automatic Control}, 
  title={Variable structure systems with sliding modes}, 
  year={1977},
  volume={22},
  number={2},
  pages={212-222},
  doi={10.1109/TAC.1977.1101446}}

@book{olver1997asymptotics,
  author    = {Frank W. J. Olver},
  title     = {Asymptotics and Special Functions},
  publisher = {A K Peters},
  address   = {Wellesley, MA},
  year      = {1997},
  isbn      = {1-56881-069-5}
}

@book{butcher2016numerical,
  author    = {Butcher, J. C.},
  title     = {Numerical Methods for Ordinary Differential Equations},
  edition   = {3rd},
  publisher = {John Wiley \& Sons, Ltd},
  year      = {2016},
  doi       = {10.1002/9781119121534}
}

@book{khalil2002nonlinear,
  author    = {Khalil, Hassan K.},
  title     = {Nonlinear Systems},
  edition   = {3rd},
  publisher = {Prentice Hall},
  address   = {Upper Saddle River, NJ},
  year      = {2002},
  isbn      = {978-0-13-067389-3}
}

@article{soderlind2002automatic,
  author    = {S{\"o}derlind, Gustaf},
  title     = {Automatic Control and Adaptive Time-Stepping},
  journal   = {Numerical Algorithms},
  volume    = {31},
  number    = {1--4},
  pages     = {281--310},
  year      = {2002},
  publisher = {Springer},
  doi       = {10.1023/A:1021160023092}
}

@article{johansson1999zeno,
  author  = {Johansson, Karl H. and Egerstedt, Magnus and Lygeros, John and Sastry, Shankar},
  title   = {On the Regularization of Zeno Hybrid Automata},
  journal = {Systems \& Control Letters},
  volume  = {38},
  number  = {3},
  pages   = {141--150},
  year    = {1999},
  doi     = {10.1016/S0167-6911(99)00059-6}
}

@book{markley2014fundamentals,
  author    = {Markley, F. Landis and Crassidis, John L.},
  title     = {Fundamentals of Spacecraft Attitude Determination and Control},
  series    = {Space Technology Library},
  volume    = {33},
  publisher = {Springer},
  address   = {New York, NY},
  year      = {2014},
  doi       = {10.1007/978-1-4939-0802-8}
}

@inproceedings{lopez2021sliding,
  author    = {Lopez, Brett T. and Slotine, Jean-Jacques E.},
  title     = {Sliding on Manifolds: Geometric Attitude Control with Quaternions},
  booktitle = {2021 IEEE International Conference on Robotics and Automation (ICRA)},
  pages     = {11140--11146},
  year      = {2021},
  doi       = {10.1109/ICRA48506.2021.9561867}
}

@article{bhat2000unwinding,
  author  = {Bhat, Sanjay P. and Bernstein, Dennis S.},
  title   = {A Topological Obstruction to Continuous Global Stabilization of Rotational Motion and the Unwinding Phenomenon},
  journal = {Systems \& Control Letters},
  volume  = {39},
  number  = {1},
  pages   = {63--70},
  year    = {2000},
  doi     = {10.1016/S0167-6911(99)00090-0}
}

@inproceedings{caubet2013optimal,
  author    = {Caubet, Albert and Biggs, James D.},
  title     = {Optimal Attitude Motion Planner for Large Slew Maneuvers Using a Shape-Based Method},
  booktitle = {AIAA Guidance, Navigation, and Control Conference},
  year      = {2013},
  note      = {AIAA 2013-4719},
  doi       = {10.2514/6.2013-4719}
}
%




%

\end{document}